\documentclass{article}
\usepackage{orcidlink}

\usepackage{authblk}

\usepackage{amssymb}
\usepackage{amsmath}

\usepackage{array}
\usepackage{url}
\usepackage{todonotes}
\usepackage{siunitx}
\usepackage{multirow}
\usepackage{comment}
\usepackage{subcaption}
\usepackage{float}
\usepackage{booktabs}

\usepackage{arxiv}
\usepackage{manfnt}
\renewcommand{\shorttitle}{Surrogate Model for Satellite Observations in NWP}
\newcommand{\ems}{\varepsilon_\nu} 

\usepackage{xspace}
\newcommand{\HMLincM}{\mathcal{H}_{\text{ML-{inc}}}}
\newcommand{\HMLinc}{$\HMLincM$\xspace}

\newcommand{\HMLfullM}{\mathcal{H}_{\text{ML-{full}}}}
\newcommand{\HMLfull}{$\HMLfullM$\xspace}

\newcommand{\HRTTOVM}{\mathcal{H}_{\text{RTTOV}}}
\newcommand{\HRTTOV}{$\HRTTOVM$\xspace}

\usepackage{siunitx}

\providecommand{\qty}[2]{\SI{#1}{#2}}
\providecommand{\unit}[1]{\si{#1}}
\providecommand{\qtyrange}[4][]{\SIrange[#1]{#2}{#3}{#4}}

\usepackage{amsthm}

\newcommand{\kg}{kg}
\newcommand{\per}{$//$}
\newcommand{\giga}{G}
\newcommand{\byte}{B}

\newtheorem{proposition}{Proposition}[section]

\theoremstyle{definition}

\theoremstyle{remark}
\newtheorem{remark}{Remark}[section]

\title{Learning Data-driven Surrogate and Correction Models for Satellite Observations in Numerical Weather Prediction}
\author[1]{Gian Luca Buono \orcidlink{0009-0003-3020-2383}}
\author[2]{Stefanie Hollborn \orcidlink{0009-0003-0016-5550}}
\author[2,3]{Roland Potthast \orcidlink{0000-0001-6794-2500}}
\author[1]{Jörg Schäfer \orcidlink{0000-0003-4797-0306}}
\author[1]{Martin Simon \orcidlink{0000-0002-1131-6364}}
\affil[1]{Department of Computer Science and Engineering, Frankfurt University of Applied Sciences, Frankfurt, Germany\cr \texttt{\{luca.buono, jschaefer, martin.simon\}@fra-uas.de}} 

\affil[2]{Meteorological Analysis and Modeling Unit, Deutscher Wetterdienst, Offenbach, Germany\cr 
\texttt{\{Stefanie.Hollborn, Roland.Potthast\}@dwd.de}} 
\affil[3]{University of Reading, United Kingdom}{}

\date{March 2026}

\begin{document}
\maketitle

\begin{abstract}
Satellite observations play a critical role in numerical weather prediction where they are assimilated through an observation operator that maps model states to radiances. In the traditional Ensemble Kalman Filter, these observations are used to update the state by weighting their associated errors against model uncertainties to produce an optimal estimate. This process requires radiative transfer simulations for passive, downward-viewing satellite radiometers operating in the visible, infrared, and microwave spectra. Typically, such simulations rely on numerically integrating physical laws via models like RTTOV. In this paper, we introduce two machine learning surrogate observation operators inspired by modern computer-vision architectures: First, \HMLfull, a fully data-driven emulator of radiative transfer, and second, \HMLinc, a hybrid incremental correction model that learns only the residual relative to RTTOV, thereby retaining established physics while enabling data-driven refinement in complex conditions such as cloud-affected situations. The residual formulation improves radiance accuracy (lower Root Mean Squared Error (RMSE) than the fully data-driven emulator and RTTOV) and adds only moderate computational costs to the assimilation step. Both models combine 3D convolutions for vertical profile encoding with a 2D U-Net operating on latitude–longitude grids, allowing joint learning of vertical structure, spatial correlations, and inter-channel dependencies.
We further provide a theoretical justification for deploying the hybrid surrogate as an observation operator in data assimilation. 
\end{abstract}

\section{Introduction}
Satellite observations are essential for numerical weather prediction (NWP), but their assimilation requires radiative transfer models (RTMs) to simulate observed radiances from atmospheric states, which can be computationally expensive \cite{bauer2015quiet, rodgers2000inverse, saunders2018update}. Traditional RTMs like RTTOV (Radiative Transfer for TIROS\footnote{Television InfraRed Observation Satellite} Operational Vertical Sounder) are physically grounded, yet they can struggle in cloud-affected (all-sky) conditions and impose a substantial computational burden on operational data assimilation \cite{bauer2010allsky, geer2018all}. Using machine learning approaches has been highly popular in many areas, including radiative transfer simulations in atmospheric sciences, e.g. \cite{krasnopolsky2010nnradiation,dueben2018challenges,chantry2021machinelearning}. In this work, we develop machine learning surrogate observation operators that either emulate satellite radiances directly or learn residual corrections to RTTOV, with the goal of improving both computational efficiency and accuracy for all-sky assimilation.

\subsection{General Context}

NWP is based on the numerical solution of the governing equations of atmospheric motion, derived from the Navier--Stokes equations under appropriate geophysical approximations. These equations describe the conservation of mass, momentum, energy, and water species and are complemented by equations of state and radiative transfer. For numerical computation, the resulting non-linear partial differential equations are discretized in space and time using finite-difference, finite-volume, or spectral methods on structured or unstructured grids. Continuous advances in numerical methods, physical parameterizations, and high-performance computing have enabled increasingly realistic atmospheric simulations across a wide range of spatial and temporal scales \cite{Kalnay2003,Durran2010,Williamson2007}.

A persistent challenge in classical NWP arises from the representation of unresolved subgrid-scale processes, including turbulence, cloud microphysics, radiation, and land-surface interactions. These processes must be parameterized and introduce systematic model error that limits forecast skill. In addition, the chaotic nature of the atmosphere leads to rapid error growth from uncertainties in initial conditions. Ensemble prediction systems were therefore introduced to quantify forecast uncertainty by sampling these uncertainties, providing probabilistic forecasts that are now an integral component of operational NWP \cite{LeutbecherPalmer2008,Palmer2019}.

Data assimilation provides the mathematical framework to optimally combine short-range forecasts with heterogeneous observations in order to estimate the atmospheric state on the model grid. While variational approaches such as 3D-Var and 4D-Var have long been used operationally, ensemble-based methods have gained increasing importance due to their ability to represent flow-dependent background-error covariances. Hybrid ensemble-variational methods (EnVAR) combine ensemble-derived covariances with variational optimization, while ensemble Kalman filter approaches, such as the Local Ensemble Transform Kalman Filter (LETKF), achieve computational efficiency through localization and ensemble-space updates \cite{Hunt2007,Bannister2017}. At convective scales, ensemble-based data assimilation has been successfully demonstrated in operational settings, for example in the kilometre-scale COSMO ensemble data assimilation system \cite{Schraff2016}. Beyond Gaussian assumptions, particle filter methods provide a fully Bayesian alternative for highly non-linear and non-Gaussian systems, and recent developments in localization strategies have enabled their application to high-dimensional atmospheric models \cite{Potthast2019}.

Alongside conventional in-situ observations from radiosondes, surface stations, and aircraft, satellite observations constitute one of the most important components of the global meteorological observing system. Satellite data provide dense spatial coverage and are routinely assimilated in numerical weather prediction across the microwave (MW), infrared (IR), and visible (VIS) spectral ranges \cite{geer2018all,eyre2022_2}. Microwave observations form the backbone of modern data assimilation systems, as they provide robust information on atmospheric temperature and humidity profiles and can be used under most cloud and precipitation conditions. Infrared observations from both polar-orbiting and geostationary satellites offer high spatial and temporal resolution and contain valuable information on temperature, humidity, and cloud properties. Historically, IR radiances were assimilated primarily under clear-sky conditions due to the strong and highly non-linear impact of clouds on radiative transfer. More recently, advances in radiative transfer modelling, quality control, and observation error specification have enabled the assimilation of cloud-affected IR radiances, opening new opportunities for improving forecasts of clouds, precipitation, and convective processes.

Visible satellite observations provide very high-resolution information on cloud structure and temporal evolution, which is particularly valuable for nowcasting and short-range forecasting of convection. The operational assimilation of VIS reflectances has progressed rapidly in recent years and has been implemented at the German Meteorological Service (DWD) since March 2023 \cite{scheck2020,jahrbuch2023}. This development highlights the increasing relevance of VIS observations for operational data assimilation.

All-sky data assimilation refers to the use of satellite observations
irrespective of cloud cover, including cloudy and precipitating conditions
\cite{geer2018all,Schomburg2026}. Traditional cloud-screening approaches, while simplifying the assimilation process, discard a large fraction of observations in dynamically active regions. Although MW all-sky assimilation is now operationally well established, all-sky assimilation of IR radiances remains an active area of research with substantial potential for further forecast improvements. Overall, all-sky approaches have been shown to reduce forecast errors, particularly for humidity, cloud, and precipitation-related variables, by increasing both the quantity and the effective information content of assimilated observations.

The use of observations in a numerical weather prediction (NWP) system by data assimilation comprises in general two key aspects: observation forward modeling and integration within the data assimilation system, with a particular emphasis on error modeling (\cite{Kalnay2003,NakamuraPotthast2015}). In data assimilation, observation operators $\mathcal{H}$ (also referred to as forward models) are required to simulate virtual observations that are consistent with the variables of a NWP model. During the assimilation process, the model state $x$ is corrected by comparing actual measured observations $y$ with their simulated counterparts. This adjustment ensures an optimal alignment of the model with real atmospheric conditions.

RTMs simulate the top-of-atmosphere radiances observed by satellite instruments, based on atmospheric profiles of temperature, moisture, trace gases, aerosols, clouds, and surface properties. Just as clouds present significant challenges for atmospheric modeling, the simulation of observations in cloud-affected conditions poses comparable complexities for observation operators. 
Recent advancements in RTMs, particularly with models like RTTOV \cite{saunders2018update, rttov14svr} and the Community Radiative Transfer Model CRTM \cite{CRTM2023}, have significantly enhanced their accuracy and computational efficiency, enabling the assimilation of data from a wide range of satellite instruments across the electromagnetic spectrum. For an overview we refer to \cite{eyre2020_1, eyre2022_2, weng2020}. These RTMs incorporate sophisticated treatments of gaseous absorption, cloud scattering, and surface emissivity--including over complex terrains like land, ice, and ocean--which has facilitated the transition toward all-sky assimilation, where cloud-affected radiances are routinely integrated into NWP systems. 

In data assimilation, observation errors are represented by the error covariance matrix $R$, which not only characterizes the variance of errors for each individual observation but also captures the correlations between errors across different observations. Historically, observation errors were assumed to be uncorrelated, resulting in a diagonal R-matrix, e.g. \cite{stewart2013data}. However, this assumption is often insufficient, particularly for high-resolution or spectrally correlated data--such as those from hyperspectral infrared or microwave sounders--where errors can exhibit significant interdependencies.
Modern approaches recognize the importance of accounting for vertically correlated errors, where adjacent channels exhibit strong interdependencies. Techniques like the Desroziers method \cite{desroziers2005diagnosis} have been instrumental in diagnosing and refining these error covariances, though practical implementation often requires conditioning to ensure numerical stability.
The transition to all-sky assimilation has necessitated scene-dependent error models, which dynamically adjust error statistics based on factors such as cloud cover, surface type, and observation geometry \cite{geer2019correlated}. Clouds introduce additional uncertainties due to forward model errors (e.g., inaccuracies in cloud representation) and increased variability in radiance measurements. To address this, error models incorporate cloud-sensitive parameters, ensuring that the assimilation system appropriately weights observations affected by clouds. While vertical error correlations are increasingly accounted for, horizontal correlations remain challenging and are often managed through observation thinning to mitigate redundancy.

\subsection{Our Contribution}
In this work, we develop a three-dimensional data-driven methodology for {\em simulating} or {\em correcting} satellite radiance measurements through surrogate observation operators $\mathcal{H}$. This means that these operators can be used both for direct observation simulation $y=\mathcal{H}(x)$ as well to calculate a correction $\delta y = \delta \mathcal{H}(x)$ to the physical simulation. Clearly, $\delta \mathcal{H}(x)$ can also be used to estimate or improve the estimation of observation errors $R = R(x)$ in NWP, which arise from various sources, including background uncertainty, instrument noise and representation errors.

The complexity of the mappings $\mathcal{H}(x)$ and $R(x)$ underscores the need for situation-dependent modeling of observation errors rather than assuming constant error characteristics. The primary goal of our work is to capture this variability using machine learning techniques to learn a fully data-driven observation operator surrogate, or the residuals with respect to standard radiative transfer models such as the RTTOV operator provided by EUMETSAT (see Section~\ref{sec:rttov}). This results in dynamic surrogate models that adapt to changing atmospheric conditions.

Machine learning techniques have previously been applied to enhance radiative transfer computations. For example, in \cite{Scheck}, a neural network was used to improve a look-up table approach, while in \cite{BaurEtAl}, a neural network was employed to simulate reflectances of the \qty{1.6}{\mu m} near-infrared channels. The authors of \cite{ZhouEtAl} proposed 3D Short-Wave radiative effects corrections to 1D models, and \cite{StegmannEtAl} introduced a deep learning method to accelerate radiative transfer simulation by improving regression coefficients within an existing model.

We propose and investigate two data-driven surrogate operators  \HMLfull and \HMLinc. While the first surrogate \HMLfull is a fully data-driven emulator attempting to map atmospheric states directly to radiances, being fully independent of RTTOV and could be henceforth considered a genuine \emph{surrogate emulator} or \emph{surrogate model}, the second surrogate operator \HMLinc adopts an incremental, hybrid formulation. That is, the (deep) network learns only the residuals with respect to RTTOV, thereby embedding established radiative-transfer physics into the surrogate while enabling data-driven refinement in complex, cloud-affected conditions. Therefore, we call it a \emph{correction model}. 

We develop a three-dimensional learning approach that incorporates both the vertical structure of atmospheric profiles and the horizontal spatial structure of the observations to approximate the observation operator or its correction. We evaluate the quality of both models demonstrating that 
the residual-learning approach not only achieves lower root-mean-square errors compared to the fully data-driven approach but also as compared to the physics only traditional numerical model. It is possible to achieve a good data-driven approximation of the observation operator, but the residual-learning approach has better quality, and it leverages decades of radiative-transfer expertise encoded in RTTOV and focuses its learning capacity on the remaining, physically meaningful deviations. To the best of our knowledge, such a three-dimensional model based on U-Net architectures has not been explored in the literature before.

\section{Methodological Background}
\label{sec:methodology}
This section establishes the methodological background for data-driven observation operators. We review the role of observation operators and their corresponding errors in data assimilation, then summarize the RTTOV radiative transfer model as a physics-based observation operator, and finally, introduce our hybrid incremental correction method, which learns residuals between RTTOV-simulated radiances and satellite observations.

\subsection{Observation Error in Data Assimilation}
Consider the state of a dynamical system $x_k \in \mathbb{R}^{n}$ at discrete times $t_k$, $k \in \mathbb{N}$. The state is propagated using a (possibly nonlinear) dynamical model, 

\begin{equation}
x_k = \mathcal{M}(x_{k-1}) + \eta_k,
\end{equation}

where $\mathcal{M}: \mathbb{R}^n\to \mathbb{R}^{n}$ denotes the forecast model and $\eta_k$ is the model error, often assumed Gaussian with zero mean and covariance $Q_k$\footnote{This assumption is not necessary for the sequel, however.}. 

The system is observed through measurement vectors $y_k \in \mathbb{R}^{m}$ related to the (true) state via the observation model

\begin{equation}
y_k = \mathcal{H}(x_k) + \epsilon_k,
\end{equation}

where $\mathcal{H}$ is the (possibly nonlinear) observation operator and $\epsilon_k$ denotes the observation error. Formally, the (unobserved) error realization satisfies 
\begin{equation}
\epsilon_k := y_k - \mathcal{H}(x_k).
\end{equation}
In many practical applications, $\epsilon_k$ is modeled as Gaussian with zero mean and covariance $R_k$. We do not impose this assumption here, as in all-sky radiance assimilation representativeness and forward-model errors often lead to non-Gaussian and state-dependent observation errors. For readability, we omit the time index $k$ from this point onward and consider a generic assimilation step.

\subsection{Contributions to Observation Error}
Observation error is commonly decomposed into instrument error, representativeness error, and observation-operator error \cite{JanjicEtAl}. Representativeness error arises because observations sample atmospheric processes at spatial and temporal scales that are unresolved or only partially resolved by the numerical model, leading to systematic mismatches between observed quantities and their model-space counterparts; this effect is strongly process dependent and can vary dynamically in space and time \cite{Vobig2021,Potthast2022,VobigPotthast2025}.

In what follows, $\;\widehat{\cdot}\;$ denotes the inaccessible \emph{true}, \emph{continuous}
quantities, in contrast to discretized ones accessible on a computer. Given a discrete state
$x$, let $x^{c}:=\mathcal I x$ denote its reconstructed (interpolated) continuous field, where
$\mathcal I$ is a prescribed reconstruction operator. We write $\widehat{\mathcal H}(x^{c})$
for the evaluation of the continuous observation operator on this reconstruction; this does
not imply that the true state $\widehat{x}$ is a function of $x$.

Assuming the observation model $y = \widehat{\mathcal H}(\widehat{x}) + \epsilon_{\text{m}}$,
define the obs-minus-model residual with respect to the discretized operator $\mathcal H$ as
\[
\epsilon := y - \mathcal H(x).
\]
Adding and subtracting $\widehat{\mathcal H}(x^{c})$ yields the decomposition
\begin{align}\label{eq:errorDecomp}
\epsilon
&= \epsilon_{\text{m}}
 + \underbrace{\widehat{\mathcal H}(\widehat{x}) - \widehat{\mathcal H}(x^{c})}_{\text{representation error }\epsilon_{\text{r}}}
 + \underbrace{\widehat{\mathcal H}(x^{c}) - \mathcal H(x)}_{\text{forward model / discretization error }\epsilon_{\text{f}}}.
\end{align}

\begin{remark} If needed, one could also include an additional \emph{pre-processing} error $\epsilon_{\text{p}}$ in (\ref{eq:errorDecomp})
but we omit it for simplicity.
\end{remark}

\subsection{The RTTOV Observation Operator}\label{sec:rttov}

Among the most widely used radiative transfer models is RTTOV (Radiative Transfer for TOVS), a fast and flexible model designed for simulating radiances from passive visible, infrared, and microwave satellite instruments. RTTOV is operationally employed at major NWP centers, including ECMWF, the Met Office, Météo-France, and DWD, and is a key deliverable of the EUMETSAT NWP SAF (Satellite Application Facility for Numerical Weather Prediction), see \cite{saunders2018update, rttov14svr}.
RTTOV simulates radiances across a broad spectral range, covering visible, infrared, and microwave wavelengths. The model accepts a comprehensive set of atmospheric input parameters, including temperature, water vapor, trace gases, aerosols, clouds, and hydrometeors. Additionally, RTTOV incorporates surface emissivity and reflection models for diverse surface types, such as land, ocean, sea ice, and snow. The model also includes scattering and extinction parameters for clouds and precipitation, thereby enabling all-sky radiance simulation.

In the sequel we denote the RTTOV operator by  $\mathcal{H}_{RTTOV}$. Note, that $\mathcal{H}_{RTTOV}$ constitutes a genuinely non-linear operator as it is effectively an integral operator integrating Plancks law under the hood.

The interactions of particles with radiation can be attributed to 
\begin{enumerate}
    \item Absorption: radiation attenuation by energetic modification (heat or chemical reaction), 
    \item Emission: isotropic increase in radiation by molecular excitation due to absorption. Note, that according to Kirchhoff’s law in thermodynamical equilibrium we have: $emission = absorption$, and 
    \item Scattering: radiation attenuation by deviation of radiation from original direction; also increase in radiation by deviation of radiation into direction under consideration.
\end{enumerate}

The following section is based on the exposition in \cite{weng2017}. Consider a pressure coordinate system extending from the surface at pressure $p = p_S$ to the top of the atmosphere at $p = 0$. Let $I_\nu(0)$ denote the radiation intensity at the top of the atmosphere ($p = 0$), and $I_\nu(p_S)$ the radiation emitted from the surface. 
Transmittance $\tau_\nu(p_1, p_2)$ between two pressure levels $p_1$ and $p_2$ is computed by applying the Beer–Lambert law, which integrates the absorption of radiation along a path through the atmosphere using gas-specific absorption coefficients and layer densities. The absorption coefficients vary significantly with temperature and atmospheric gas concentrations.

Assuming local thermodynamic equilibrium, the radiation at frequency $\nu$ reaching the satellite is given by the radiative transfer equation:

\[
I_\nu(0) = I_\nu(p_S) \, \tau_\nu(p_S, 0) + \int_{p_S}^{0} B_\nu \big( T(p) \big) \frac{\partial \tau_\nu(p, 0)}{\partial p} \,\mathrm{d}p,
\]
where $T(p)$ is the temperature at pressure $p$, and $B_\nu(T)$ is the Planck function at frequency $\nu$ and temperature $T$.

The radiation emitted from the surface depends on the surface temperature $T_S$, the surface emissivity $\ems$, and the reflected atmospheric and solar radiation. It can be expressed as:

\[
I_\nu(p_S) = \ems B_\nu(T_S) + (1 - \ems) \int_{0}^{p_S} B_\nu \big( T(p) \big) \frac{\partial \tau_\nu(p_S, p)}{\partial p} \,\mathrm{d}p.
\]

In the original RTTOV formulation the atmospheric layer was assumed to be optically thin corresponding to so-called clear-sky situation, i.e.\ without clouds. For optically thick layers (e.g., with cloud) only the upper regions of the atmosphere provide a significant contribution to the radiance, which has to be taken into account by suitable modifications of the calculations outlined above, for further details we refer to \cite{weng2017} and \cite{saunders2018update}.

The radiative transfer equations must be adapted to the observation geometry, since the path length and direction of radiation through the atmosphere depend on both the solar zenith angle (for incoming solar radiation) and the satellite viewing angle. When the satellite observes the atmosphere at an oblique angle rather than directly, the effective path through each atmospheric layer is longer, and in some cases the satellite may not observe radiation from the true top of atmosphere (e.g., due to limb geometry or partial atmospheric blocking), which must be accounted for in the transmittance and radiance calculations.

In RTTOV, atmospheric transmittance is computed using precomputed optical depth coefficients derived from high-resolution line-by-line radiative transfer simulations. For each channel, RTTOV interpolates these coefficients based on the input atmospheric profile (temperature, gas concentrations, pressure levels), then integrates layer-by-layer optical depths along the specified viewing path to obtain transmittance values. This coefficient-based approach allows RTTOV to rapidly compute accurate transmittances suitable for operational data assimilation.
The surface emissivity can either be supplied directly by the user or calculated internally using built-in models and emissivity atlases. For infrared channels and land or sea-ice surfaces, RTTOV uses dedicated IR emissivity models or spectral atlases; if no detailed information is provided, default emissivity values are applied.

In operational numerical weather prediction the computational efficiency of satellite data assimilation is critical, as real-time data assimilation demands rapid processing of vast observational datasets. While RTTOV is optimized for speed, further reductions in computational load are essential to maintain operational feasibility. A common approach to achieving this involves simplifying observation error modelling, where certain error correlations--such as inter-channel, horizontal, and temporal correlations--are often neglected. Although this reduces computational overhead, it can lead to suboptimal use of observational information, potentially degrading forecast accuracy.
To address this limitation, data-driven approaches present a promising alternative. Training can be performed offline, allowing the model to learn complex error structures in the data. By capturing error dependencies that are typically neglected in traditional error models, data-driven approaches could enhance the information content extracted from observations, ultimately improving the accuracy and robustness of data assimilation systems. In order to learn such error dependencies from data, one has to choose an ML model architecture capable of learning the observation operator. Ideally, the ML model architecture is capable of learning the spatial, temporal and channel correlations. To this end, we need to carefully understand the spatial structure both from a theoretical and modeling perspective as it influences the estimation of errors as well as from an operational perspective as the data structures of our ML model depend on it. 

It is important to realize that the dichotomy between clear-sky and all-sky situations essentially introduces a binary source of uncertainty as two different numerical models of the observation error are applied depending on the (unknown) context clear- or all-sky. As the difference between clear-sky and all-sky estimates is significant, the numerical estimation of the observation error requires a sufficiently sophisticated identification of the underlying scene, i.e., correctly determining whether the situation corresponds to clear-sky or all-sky conditions. Our data-driven approach on the other hand does not distinguish between clear- and all-sky regimes so there is no need for providing a  sufficiently sophisticated \lq\lq guess\rq\rq\ of the scene. 

\subsection{Incremental Hybrid Correction Approach}
\label{subsec:datadrivenObsOpSurrogate}
RTTOV is the established physics-based forward model for radiative transfer in operational data assimilation. In this work, we develop an incremental hybrid surrogate, denoted by \HMLinc, in which the data-driven component does not replace RTTOV but learns an additive correction to its output. Concretely, we preserve RTTOV’s input–output structure and model 
\begin{equation}\label{eq:hmlinc_def}
\HMLincM(x) := \HRTTOVM(x) + \delta_{\mathrm{ML}}(x),
\end{equation}
where $\delta_{\mathrm{ML}}(x)$ is a neural network trained on residuals between RTTOV-simulated radiances and satellite observations.

This incremental hybrid construction enables diagnostic insight into the relative magnitude of the error contributions in \eqref{eq:errorDecomp} under cloudy and clear-sky conditions: if a deterministic component (e.g., forward-model error) dominates in certain regimes, it is natural to learn an associated, state-dependent bias correction from data.

\noindent 
Let
\[
\delta(x) := \widehat{\mathcal{H}}(x^c) - \HRTTOVM(x)
\]
denote the \emph{operator mismatch} of the true operator caused by discretization and using RTTOV.
If $\HMLincM(\cdot)$ is a good approximation of the (inaccessible) true observation operator, then $\delta_{\text{ML}}(\cdot)$ provides an approximation of $\delta(\cdot)$. This is the key idea of our approach: the surrogate makes an estimation of the  inaccessible operator mismatch \emph{computable}. More precisely, observations consist of pairs $(x,y)$ satisfying
\[
y = \widehat{\mathcal H}(x^c) + \epsilon_{\mathrm m},
\]
note, however, that $\widehat{\mathcal H}(x^c)$ is not computable and consequently $\delta(x)$ can not be identified pointwise.

What is identifiable under the assumption $\mathbb E[\epsilon_{\mathrm m}\mid x]=0$ is the conditional bias function
\[
b(x) := \mathbb E\bigl[y-\HRTTOVM(x)\mid x\bigr].
\]
Under the model $y=\HRTTOVM(x)+\delta(x)+\epsilon_{\mathrm m}$, this conditional bias satisfies $b(x)=\delta(x)$. In this sense, the learned correction $\delta_{\text{ML}}(\cdot)$ targets the systematic discrepancy between observations and the baseline operator $\HRTTOVM(x)$ conditional on the model state $x$. It should be emphasized that $\delta(x)$ in this formulation does not represent a pure discretization error in an absolute physical sense. Rather, it may aggregate multiple sources of systematic discrepancy, including forward-model bias, representativeness effects, and any remaining state-dependent preprocessing biases present in the observations. A clean physical separation of these contributions would require additional structural assumptions beyond those considered here.

As a result, the innovation (observation residual) admits the exact decomposition
\begin{align*}
\epsilon
&:= y-\HRTTOVM(x) \\
&= \underbrace{y-\HMLincM(x)}_{\text{residual under surrogate}}
   + \underbrace{\HMLincM(x)-\HRTTOVM(x)}_{\delta_{\mathrm{ML}}(x)}.
\end{align*}
If $\HMLincM$ accurately captures the systematic component of the observation discrepancy, i.e.,\ the conditional bias
$b(x)=\mathbb E[y-\HRTTOVM(x)\mid x]$, then the residual term $y-\HMLincM(x)$ is dominated by the remaining non-systematic components of the observation error, including measurement noise and, in practice, representativeness effects. In this sense,
\begin{equation}
\epsilon \approx \epsilon_{\mathrm m} + \delta_{\mathrm{ML}}(x),
\label{eq:errorDecomp3}
\end{equation}
where the approximation is understood in terms of conditional moments rather than pointwise equality.

That is, $\delta_{\text{ML}}(x)$ acts as a state-dependent, deterministic bias term that is available during the EnKF cycle, whereas $y-\HMLincM(x)$ primarily captures random components.
Therefore, our incremental hybrid approach learns \emph{residual corrections} rather than the observations directly; this can be beneficial as is proved in Proposition \ref{prop:residual_better} below. 
\begin{remark}
Note that equivalently, once $\delta_{\text{ML}}$ is available, one may define a bias-corrected observation $y^{\text{corr}} := y - \delta_{\text{ML}}(x)$, for which the innovation satisfies
$y^{\text{corr}} - \HRTTOVM(x) = y - \HMLincM(x)$. Thus, using the surrogate operator is equivalent to assimilating a bias-corrected observation while retaining the baseline RTTOV operator.
\end{remark}

Let us specify assumptions under which an incremental surrogate \HMLinc can approximate the true observation operator. 
Unless stated otherwise, all expectations are taken with respect to the joint distribution of the random variables induced by the data-generating process; we write $\mathbb E[\cdot\mid x]$ for conditional expectation given $x$ and $\mathbb E_x[\cdot]$ when emphasizing integration over the marginal of $x$. We assume the measurement error has zero conditional mean, 
\begin{equation*}
\mathbb{E}[\epsilon_{\text{m}} \mid x]=0,
\end{equation*}
and training data $\mathcal{D}=\{(x^{(i)},y^{(i)})\}_{i=1}^N$ satisfy
\begin{equation*}
y^{(i)}=\widehat{\mathcal{H}}((x^{(i)})^c)+\epsilon_{\text{m}}^{(i)},\quad i=1,...,N.
\end{equation*}
We define residual targets
\begin{equation*}
r^{(i)} := y^{(i)}-\HRTTOVM(x^{(i)}) 
= \bigl(\widehat{\mathcal{H}}((x^{(i)})^c)-\HRTTOVM(x^{(i)})\bigr)+\epsilon_{\text{m}}^{(i)}.
\end{equation*}
We then learn a correction $\delta_{\text{ML}} \in \mathcal{G}$, where $\mathcal{G}$ denotes a neural network hypothesis class, via
\begin{equation*}
\min_{f \in \mathcal{G}}\frac{1}{N}\sum_{i=1}^N \bigl\|f(x^{(i)})-r^{(i)}\bigr\|^2.
\end{equation*}
By the law of total expectation, we have for the loss functional
\[
\mathcal{L}(f)=\mathbb{E}\bigl[\|f(x)-r\|^2\bigr]
=\mathbb{E}_{x}\Bigl[\mathbb{E}\bigl[\|f(x)-r\|^2 \mid x\bigr]\Bigr],
\]
and minimizing the inner conditional expectation pointwise in $x$ shows that squared loss targets the conditional mean,
\begin{equation*}
\delta_{\text{ML}}(x)=\mathbb{E}[r\mid x]
=\delta(x) = \widehat{\mathcal{H}}(x^c)-\HRTTOVM(x),
\end{equation*}
assuming sufficient expressiveness of $\mathcal{G}$. Hence the incremental surrogate
\begin{equation*}
\HMLincM(x)=\HRTTOVM(x)+\delta_{\text{ML}}(x)
\end{equation*}
targets the conditional mean $\mathbb E[y\mid x]=\widehat{\mathcal{H}}(x^c)$, with the remaining error governed by the irreducible conditional variance of $\epsilon_{\mathrm m}$. That is, the remaining random component is governed by the generally heteroscedastic covariance
\begin{equation*}
R(x):=\operatorname{Cov}(\epsilon_{\text{m}}\mid x),
\end{equation*}
which can be estimated in a second stage from the conditional second moment of the innovation $y-\HMLincM(x)$ (e.g., via a separate regression model for $R(x)$ or via regime-dependent binning such as clear-sky vs.\ all-sky).

To formalize why the hybrid incremental approach of learning residual corrections can be advantageous, we compare fully data-driven and hybrid incremental learning under squared loss and isolate the respective approximation errors.
\begin{proposition}
\label{prop:residual_better}
Assume the observation model
\begin{equation*}
y = \widehat{\mathcal H}(x^c) + \epsilon_{\mathrm m},
\qquad \mathbb{E}[\epsilon_{\mathrm m}\mid x]=0,
\end{equation*}
and define the operator mismatch relative to some baseline operator $\mathcal H$
\begin{equation*}
\delta(x):=\widehat{\mathcal H}(x^c)-\mathcal H(x),
\quad\text{so that}\quad
y=\mathcal H(x)+\delta(x)+\epsilon_{\mathrm m}.
\end{equation*}

Let $\mathcal G$ be a hypothesis class and define the best-in-class direct and residual predictors
\begin{align*}
&\HMLfullM(x), &&\text{where } \HMLfullM\in \arg\min_{f\in\mathcal G}\mathbb E\bigl[\|f(x)-y\|^2\bigr],\\
&\HMLincM(x):=\mathcal H(x)+\delta_{\text{ML}}(x), &&\text{where }\delta_{\text{ML}} \in \arg\min_{f\in\mathcal G}\mathbb E\bigl[\|f(x)-(y-\mathcal H(x))\|^2\bigr].
\end{align*}

Define the approximation errors
\begin{align*}
\varepsilon_{\text{ML-full}}:= \inf_{f\in\mathcal G}\mathbb{E}\!\left[\|f(x)-(\mathcal H(x)+\delta(x))\|^2\right],\qquad
\varepsilon_{\text{ML-inc}}:= \inf_{f\in\mathcal G}\mathbb{E}\!\left[\|f(x)-\delta(x)\|^2\right].
\end{align*}
Then
\begin{align}
\mathbb E\!\left[\|\HMLfullM(x)-y\|^2\right]
&= \varepsilon_{\text{ML-full}} + \mathbb E_x\!\left[\mathbb E\bigl[\|\epsilon_{\mathrm m}\|^2\mid x\bigr]\right],\label{eqn:dir}\\
\mathbb E\!\left[\|\HMLincM(x)-y\|^2\right]
&= \varepsilon_{\text{ML-inc}} + \mathbb E_x\!\left[\mathbb E\bigl[\|\epsilon_{\mathrm m}\|^2\mid x\bigr]\right]\label{eqn:indir}.
\end{align}
\end{proposition}

\begin{proof}
Using $y=\mathcal H(x)+\delta(x)+\epsilon_{\mathrm m}$ and $\mathbb E[\epsilon_{\mathrm m}\mid x]=0$,
\begin{align*}
\mathbb E\bigl[\|f(x)-y\|^2 \mid x\bigr]
&=\|f(x)-(\mathcal H(x)+\delta(x))\|^2 + \mathbb E\bigl[\|\epsilon_{\mathrm m}\|^2\mid x\bigr],
\end{align*}
since the cross term vanishes conditionally. Taking expectations over $x$ yields
\[
\mathbb E\bigl[\|f(x)-y\|^2\bigr]=\mathbb E\bigl[\|f(x)-(\mathcal H(x)+\delta(x))\|^2\bigr]
+ \mathbb E_x\bigl[\mathbb E[\|\epsilon_{\mathrm m}\|^2\mid x]\bigr].
\]
Minimizing over $f\in\mathcal G$ gives the first identity with $\varepsilon_{\text{ML-full}}$.

For hybrid incremental learning, let $f$ approximate $y-\mathcal H(x)=\delta(x)+\epsilon_{\mathrm m}$. The same argument gives
\[
\mathbb E\bigl[\|f(x)-(y-\mathcal H(x))\|^2 \mid x\bigr]
=\|f(x)-\delta(x)\|^2+\mathbb E[\|\epsilon_{\mathrm m}\|^2\mid x],
\]
and adding back $\mathcal H(x)$ does not change the noise term, yielding the second identity.
\end{proof}

\begin{remark}
In many applications, the operator mismatch $\delta(x)$ is empirically ``simpler'' than the full mapping
$\widehat{\mathcal H}(x^c)$ (e.g., smaller amplitude, smoother spatial structure, or closer to a linear response around the baseline operator).
For common approximation families, simpler targets typically admit smaller best-in-class $L^2$ approximation error, so one often observes
$\varepsilon_{\text{ML-inc}} \ll \varepsilon_{\text{ML-full}}$ and therefore (\ref{eqn:indir}) $\ll$ (\ref{eqn:dir}). It turns out that this is also the case here, see Section 5 below.
\end{remark}

\section{Data}\label{sec:seviri}
SEVIRI is a multispectral imager onboard the Meteosat Second Generation satellites that observes the Earth in visible, near-infrared, and thermal infrared channels, providing key information on surface and atmospheric properties. This study uses three weeks of data from 10 SEVIRI channels on METEOSAT-10, combining observed radiances with numerical weather prediction data from the ICON model, i.e.\ the NWP used by DWD (\cite{Zaengl2015, icon-model-website}), and simulated radiances generated by the RTTOV radiative transfer model. The dataset includes detailed atmospheric profiles, surface variables, viewing geometry, and metadata, enabling joint analysis of measured, modeled, and simulated observations and is described in detail in the sequel.

\subsection{Data Description}
\label{sec:dataDescription}
The Spinning Enhanced Visible and Infrared Imager (SEVIRI) is a key instrument onboard the Meteosat Second Generation (MSG) satellites operated by EUMETSAT. SEVIRI is designed to observe the Earth in 12 spectral channels, allowing it to provide detailed imagery. These channels include both infrared and visible light bands, with eight channels dedicated to thermal infrared, which are essential for measuring the temperatures of clouds, land, and sea surfaces. Furthermore, certain channels are specifically sensitive to water vapour providing detailed information on the composition of the atmosphere, see \url{https://space.oscar.wmo.int/instruments/view/seviri}. 

The study utilizes data from 10 channels of the SEVIRI instrument onboard METEOSAT-10. Over a three week period, satellite data, comprising measured radiances and reflectances, have been gathered. These data files (e.g., \texttt{satellite\_data\_202401251200.nc}) have undergone a data curation process that integrates them with NWP Model data from the ICON model, including atmospheric columns of temperature, pressure, humidity, and cloud variables, see Table~\ref{tbl:data}. Furthermore, these are combined with simulated satellite data produced by RTTOV software (version $13.2.0$), along with essential metadata such as time and location, satellite zenith angle and sun zenith angle listing all available data and its type.
The table's last column ({\em{source}}) indicates the origin of the data, namely ICON if the data originates from the NWP model, SAT if the data corresponds to satellite measurements and is obtained from the satellite data provider, and RTTOV if they are derived by applying the observation operator RTTOV to the NWP data for simulating the measurements.

\begin{table}[h]
	\centering
	\caption{Spectral Channels of SEVIRI used}
	\label{tbl:channels}
	\begin{tabular}{llll} 
		\toprule
		\textbf{Channel-type} & \textbf{Wavelength [\unit{\micro\meter}]} & \textbf{Unit} & \textbf{Channel ID} \\
		\midrule
		VIS & 0.6 & \unit{1} & 0\\
		VIS & 0.8 & \unit{1} & 1\\
		Near IR & 1.6 & \unit{1} & 2\\
		IR & 3.9 & \unit{K} & 3\\
		IR & 6.2 & \unit{K} & 4\\
		IR & 7.3 & \unit{K} & 5\\
		IR & 8.7 & \unit{K} & 6\\
		IR & 9.7 & \unit{K} & 7\\
		IR & 10.8 & \unit{K} & 8\\
		IR & 12.0 & \unit{K} & 9\\
		\bottomrule
	\end{tabular}
\end{table}

\begin{table}[!ht]
	\center
	\caption{Data Description}
	\label{tbl:data}
	\newcolumntype{L}{>{\arraybackslash}m{3.7cm}}
	\begin{tabular}{|l|l|L|l|l|l}
		\hline
		\textbf{Name} & \textbf{Shape} & \textbf{Description} & \textbf{Units} & \textbf{Source}\\
		\hline
		chan & 10 & satellite channel & -- & SAT\\
		lon & 96326 & longitude & \unit{\deg} & SAT / ICON \\
		lat & 96326 & latitude &  \unit{\deg} & SAT / ICON\\
		t & (96326, 68) & temperature & \unit{K} & ICON\\
		p & (96326, 68) & pressure & \unit{hPa} & ICON\\
		qv & (96326, 68) & specific humidity & \unit{\kg/kg} & ICON\\
		qc\_dia & (96326, 68) & diagnostic cloud water & \unit{g/m^3} & ICON \\
		qi\_dia & (96326, 68) & diagnostic cloud ice & \unit{g/m^{3}} & ICON\\
		cfrac & (96326, 68) & cloud fraction & 1 & ICON\\
		dens & (96326, 68) & total air density &  \unit{kg/m^3} & ICON\\
		rad\_allsky & (96326, 10) & reflectance, brightness temperature with clouds & \unit{1} or \unit{K} & ICON / RTTOV\\
		rad\_clearsky & (96326, 10) & reflectance, brightness temperature with no clouds & \unit{1} or \unit{K} & ICON / RTTOV\\
		obs\_rad & (96326, 10) & observed reflectance, brightness temperature by SEVIRI & \unit{1} or \unit{K} & SAT\\
		t\_surf & 96326 & surface skin temperature & \unit{K} & ICON\\
		p\_surf & 96326 & surface pressure & \unit{hPa} & ICON\\
		t2m & 96326 & \qty{2}{m} temperature & \unit{K} & ICON\\
		q2m & 96326 & \qty{2}{m} specific humidity & \unit{\kg\per \kg} & ICON\\
		u10m & 96326 & \qty{10}{m} zonal velocity & \unit{m/s} & ICON\\
		v10m & 96326 & \qty{10}{m} meridional velocity & \unit{m/s} & ICON\\
		watertype & 96326 & watertype: fresh, ocean & 1 & ICON\\
		stype & 96326 & surface type: land, sea, sea-ice & 1 & ICON\\
		sun\_zen & 96326 & sun zenith angle &  \unit{\deg} & SAT\\
		sun\_azi & 96326 & sun azimuth angle &  \unit{\deg} & SAT\\
		sat\_zen & 96326 & satellite zenith angle &  \unit{\deg} & SAT\\
		sat\_azi & 96326 & satellite azimuth angle &  \unit{\deg} & SAT\\
		\hline
	\end{tabular}
\end{table}

The measurements of the different SEVIRI channels are highly correlated, see \cite{Dance2016}. Our own analysis confirms this and an example of the correlation structure is depicted in Fig.~\ref{fig:corr}. 
A clear block structure is discernible. The first three channels, located in the visible and near-infrared ranges, exhibit high correlation, as do the seven infrared (IR) channels. Additionally, the near-infrared (NIR) channel displays correlations extending into the infrared range.

\begin{figure}[!ht]
    \centering
    \includegraphics[width=0.6\linewidth]{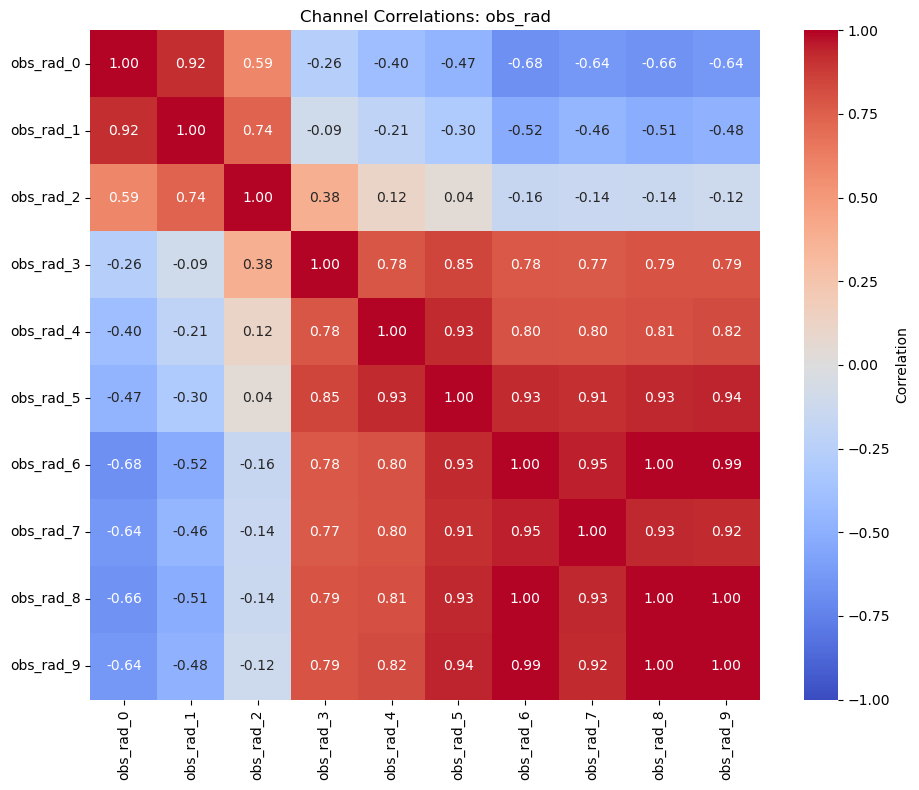}
    \caption{Channel correlations of observed reflectance (\texttt{obs\_rad})}
    \label{fig:corr}
\end{figure}

The observation data $y_k$ does not only carry a temporal structure but also an important spatial one as every data point measured is tagged with the \texttt{lat}, and \texttt{lon} coordinates of the corresponding point on the planet's surface. 

\subsection{RTTOV Calculation and Relation to Observation Error}
The observation operator $\mathcal{H}_{RTTOV}$ in our context is comprised of the Radiative Transfer for TOVS (RTTOV) model as described in Section~\ref{sec:rttov}. As indicated in Table~\ref{tbl:data} it is not necessary to calculate it as part of the training process as it is provided by ICON as part of the SEVIRI data files which already include \texttt{rad\_allsky}, and \texttt{rad\_clearsky}. We remind that $\HMLincM(x) := \HRTTOVM(x) + \delta_{\text{ML}}(x)$. Therefore, we learn \HMLinc by training the model to predict the residuals between the RTTOV simulations and the observed satellite data.

\section{Implementation}
In the following, we describe the implementation of our machine learning framework, which was applied to realize \HMLfull and \HMLinc. A fundamental difference from radiative transfer models that ignore spatial correlations and perform a (vertical) column integration is that both of our models intrinsically operate on two-dimensional spatial maps and thus appropriately account for spatial correlations. To this end, we map spatially irregular satellite observations onto a regular Cartesian grid to enable efficient image-based learning, and then optionally remap the results back to the original satellite coordinates. Atmospheric profile data are encoded using a convolutional autoencoder and processed by a 2D U-Net, treating the task as an image-to-image transformation problem. In addition, both operators learn all profiles and all channels (see Section~\ref{sec:dataDescription} and Table~\ref{tbl:data}) simultaneously and therefore, unlike RTTOV, take correlation structures into account. 

The following sections describe the mapping, the data types and the data pipeline, the architecture of the ML model and its training procedure in detail.
\label{sec:rttovEmulator}

\subsection{Spatial Data}
Satellite observations arrive as irregularly spaced measurements along the instrument's scanning geometry and must be mapped to a regular grid structure suitable for neural network processing and data assimilation. The transformation to a regular Cartesian grid is essential for leveraging convolutional neural networks, which exploit spatial correlations through their translation-equivariant architecture. This section describes our approach to handling the spatial structure of SEVIRI data, including the forward mapping from raw satellite coordinates to a regular grid (Raw2Grid) and the reverse interpolation back to original observation locations (Grid2Raw). These transformations enable efficient computation while preserving the ability to evaluate model performance in the satellite's native coordinate system.

\subsubsection{Mapping: Raw2Grid}
\label{sec:raw2grid}
The spatial structure of the data is determined by the satellite’s raw scan geometry. To obtain a manageable representation for machine learning and data assimilation, we apply a thinning preprocessing step similar to the procedure used in operational DWD data assimilation. This reduces the raw observations to a standardized grid, as described in detail below.

The satellite data contain arrays of latitude and longitude positions that are not evenly spaced and therefore do not form an equidistant grid in Cartesian coordinates. In addition, the raw satellite data are typically thinned prior to assimilation, since the full observation set is too dense for efficient computation in operational data assimilation systems. In principle, irregular spatial data could be handled using graph-based representations and graph neural networks. However, since the observation geometry is relatively regular and missing values are rare, we adopt a simpler approach in this work. Specifically, we map the observations onto a regular grid as described below. More sophisticated mappings will be explored in future work.

We applied the following procedure and call it Raw2Grid (Raw to Grid):
\begin{enumerate}
\item The grid is mapped to a--if necessary down-sampled--equidistant (in a cartesian coordinate system) spaced two dimensional grid.
\item Both the state data $x_k$ and the observation data $y_k$ are mapped onto this grid.
\item If necessary some simple interpolation is used.
\end{enumerate}

This provides the following advantages
\begin{enumerate}
\item It is conceptually simple.
\item Any data is mapped to data that can be displayed (and regarded) as an image which has several advantages for debugging and for the ML data pipeline.
\item The grid can be tailored to meet the available resources.
\end{enumerate}

\subsubsection{Reverse Mapping: Grid2Raw}
\label{sec:grid2raw}
We do not consider only the mapping of the satellite's raw data to a (smaller) grid as explained above but also investigate the \emph{inverse} mapping from the grid back to the original data (coordinates) and call it \emph{Grid2Raw} (Grid to Raw):

The motivation for reversing the mapping is threefold: First of all, as the mapping from the satellite to the mapped grid coordinate system is a lossful transformation, it is instructive to compare the results in the original coordinate system of the satellite, i.e., in the coordinate system of the raw data and \emph{not} only in the mapped grid coordinate system, to assess the impact of the lossful transformation on the estimation error. Secondly, integrating the reverse mapping into the loss function enables to minimize RMSE directly on the raw data during learning. Finally, implementing and validating the reverse mapping constitutes a first step toward a comprehensive stability analysis with respect to discretization (grid invariance), which we plan to address in future work.

A satellite coordinate is defined by its \texttt{lat} and \texttt{lon} numbers. Therefore, in order to determine the mapping back to the satellite coordinate system one just has to map the values from the (regular) grid to an arbitrary point inside the grid. This is a well-known multivariate interpolation problem and can be solved in many ways.
We decided to use bilinear interpolation implemented via PyTorch's \textit{grid\_sample} method. The multivariate interpolation we implemented extends standard bilinear interpolation methods as described in \cite{WeiserEtAl} and handles the simultaneous interpolation of multiple variables and channels.
For computational efficiency, we employ a hybrid approach that uses direct grid lookup and bilinear interpolation to map from the \texttt{[175×214]} grid back to \num{96326} scattered/raw target points. The bilinear interpolation normalizes coordinates to the \texttt{[-1, 1]} range required by \textit{grid\_sample} and processes all variables and channels simultaneously.

\subsection{Data Handling}
Our entire training dataset contains 500 NetCDF files, each containing RTTOV inputs and outputs and corresponding satellite observations from Meteosat-10 with associated metadata. Each file corresponds to one hour of satellite observations, resulting in a total time span of 3 weeks.
We split the dataset into training, validation and testing subsets, separating each subset into distinct folders before preprocessing as described in Section~\ref{sec:raw2grid} to prevent data leakage.
We used a \qty{70}{\percent} training, \qty{20}{\percent} validation, and \qty{10}{\percent} testing split. As the satellite observations, specifically the first three channels (see Table~\ref{tbl:channels}) contain NaN values, which are mapped to very large values of magnitude $10^{30}$ when loaded using \textit{xarray}. We identified these large values and replaced them with NaN values. Subsequently, we implement NaN-aware loss functions that ignore any occurrence of NaN values during training. These invalid values appear primarily during nighttime when insufficient light is reflected from Earth's surface, indicating ``no measurement'' conditions and justifying their exclusion from the loss computation.

\subsection{Model architecture}
To select a suitable model architecture, we formulate the learning problem as an image-to-image transformation applied independently at each time step, without an explicit temporal sequence model.
U-Nets were introduced by Ronneberger et al.~\cite{Unet} to support bio-medical image segmentation. A U-Net's architecture consists of a contracting path to capture image context and a symmetric expanding path that synthesizes the desired output. This architecture fits well into our use case as the input data consists of images of physical data provided in the assimilation process by the model and of output images of reflectance, and brightness temperature provided as measured observation data by SEVIRI. The architecture and data flow used in our implementation are described below.

\subsubsection{(3D) Convolutional Autoencoder}
To integrate the multi-variable atmospheric profile data, we implement a convolutional autoencoder that reduces the 14 atmospheric variables to a unified learned representation while preserving the spatial grid structure and vertical profile information. 
Using neural networks to reduce the dimensionality of (input) data is a rather old idea going back to \cite{HintonEtAl}. In the past decade instead of plain neural networks, convolutional neural networks have been proposed to encode and decode the data, primarily within the context of convolutional autoencoders, see e.g.~\cite{MasciEtAl} or \cite{BadrinarayananEtAl}. In our architecture, we only use the encoder part of an autoencoder which processes input tensors of shape \texttt{[batch, 175, 214, 14, 68]} representing batch size, longitude, latitude, atmospheric variable and height profile, respectively. The input data is passed through a series of four convolution layers each reducing the size progressively from an input size of 14 to 8, 4, 2, and eventually to 1 dimension while maintaining the spatial and vertical dimensions.

Each convolution uses kernel size $(5,5,3)$ with padding $(2,2,1)$, enabling the network to learn patterns across $5\times5$ spatial neighborhoods in the longitude-latitude grid and 3-level windows in the vertical profile dimension. The first three convolutional-layers are each followed by batch normalization, ReLU activation, and dropout. Fig.~\ref{fig:VarEncUnet} shows how the Convolutional Autoencoder is integrated into our architecture to reduce the dimensionality of our input data.

Through this design, the encoder transforms the input samples from 4D (\texttt{[batch, 175, 214, 14, 68]}) to 3D (\texttt{[batch, 175, 214, 68]}), effectively learning a single integrated representation that captures the relationships between multiple atmospheric variables at each spatial location and pressure level. The 3-dimensional encoded samples can then be processed in an image-to-image manner through our U-Net. 

\subsubsection{U-Net Architecture and Implementation}
The 3-dimensional encoded atmospheric profiles are processed through a 2D-U-Net that treats each spatial location as a pixel with 68 feature channels, enabling image-to-image style processing of the atmospheric data.
As proposed by \cite{Unet} we follow the encoder-decoder architecture, with four encoding and four decoding blocks, skip-connections and a bottleneck between encoding and decoding path.
Each encoder block implements a double convolution, $3 \times 3$ convolutions, $2 \times 2$ max-pooling and batch normalization. 
In addition to the skip-connections between encoding and decoding blocks as of the seminal paper \cite{Unet}, we decided to introduce residual connections in the convolution blocks of the decoder to improve gradient flow \cite{recurrent-residual-cnn}.
For the decoder we used 2-dimensional transposed convolutions to perform learned up-sampling in each decoder block. We chose to use transposed convolutions over interpolation-based up-sampling after observing that the latter produced heavy artifacts on our outputs/reconstructions. 
Each transposed convolution is followed by a residual convolution block consisting of two convolutions with ReLU activation, batch normalization and dropout. Due to the irregular shape ($175 \times 214$) of our input we incorporated additional interpolation to ensure that the up-sampled feature maps match the dimensions of the corresponding encoder feature maps prior to concatenation via skip connections. 
The output layer consists of a $3 \times 3$ convolution with ReLU activation, followed by reflection padding and another convolution with a $1 \times 1$ kernel that maps the 64 features to the target number of output channels. 

\begin{figure}[!ht]
    \centering
    \includegraphics[width=\linewidth]{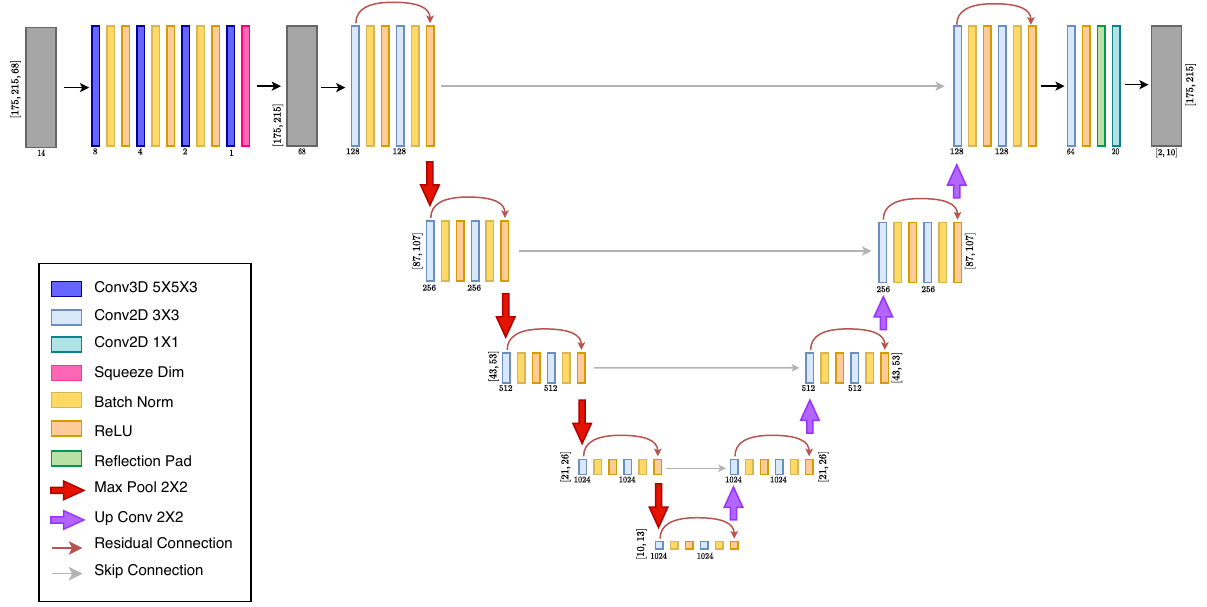}
    \caption{Our U-Net-based implementation including a convolutional autoencoder.}
    \label{fig:VarEncUnet}
\end{figure}


\subsubsection{Two Observation Operators \HMLfull and \HMLinc}
We develop two different  observation operators \HMLfull and \HMLinc to evaluate against the established RTTOV model:
\begin{enumerate}
        \item \HMLfull: This operator performs direct learning of the observed (\texttt{obs\_rad}) values. The model predicts all 10 channels of \texttt{obs\_rad} over the entire grid, which results in an output shape of \texttt{[batch\_size, 175, 214, 10]}, see Section~\ref{sec:inoutdata}. This corresponds to a fully data-driven surrogate of the RTTOV model, trained directly on observed radiances without requiring RTTOV simulations in the training loop.        
        \item \HMLinc: This operator performs residual learning using two target residuals, \texttt{obs\_rad-rad\_allsky} and \texttt{obs\_rad-rad\_clearsky}, corresponding to an incremental correction formulation. We consider two target residuals because corrections can be learned with respect to both RTTOV clear-sky and RTTOV all-sky simulations.
        
         The model predicts the residuals for all 10 channels of \texttt{rad\_allsky} and \texttt{rad\_clearsky} for the entire grid, resulting in an output shape of \texttt{[batch\_size, 175, 214, 2, 10]}, see Section~\ref{sec:inoutdata}. In order to evaluate our incremental correction model against pure RTTOV, we reconstruct \texttt{obs\_rad} from our learned residuals and compare it to the target \texttt{obs\_rad} from the SEVIRI data. We perform this reconstruction by adding our predicted residuals to the RTTOV outputs \texttt{rad\_allsky} and \texttt{rad\_clearsky} consistent with Eq.~\ref{eq:hmlinc_def}. \\
Note, that this operator does not need any RTTOV model calculations, as one has to use these only in order to calculate \texttt{obs\_rad} back from the residuals.This yields two reconstructions of \texttt{obs\_rad}--one corresponding to each RTTOV output--which we then compare against their respective RTTOV references in both grid and raw formats. 
\end{enumerate}    
The input $\mathbf{x}_k \in \mathbb{R}^{H \times W \times n_v \times n_p}$ is a tensor where $n_v = 14$ atmospheric variables (see Table~\ref{tbl:inputVariables}), $n_p = 68$ vertical profile levels, and $H \times W = 175 \times 215$ denote the spatial grid dimensions after our mapping, as described in Subsection \ref{sec:raw2grid}. 
For training \HMLfull, the target output $\mathbf{y}_k \in \mathbb{R}^{H \times W \times n_c}$ comprises $n_c = 10$ simulated satellite channels over the same spatial grid. For training \HMLinc, the target additionally includes a mode dimension for RTTOV outputs (clear-sky and all-sky), yielding 
$\mathbf{y}_k \in \mathbb{R}^{H \times W \times n_m \times n_c}$ with $n_m = 2$. In \texttt{PyTorch} parlance, both tensors are effectively \emph{multi-channel images} of spatial extent $H \times W$.

We optimize both models by minimizing the Mean Squared Error MSE or RMSE between the model's predictions and the target satellite observations, respectively the target residual values.
           
\subsection{Data}
\label{sec:inoutdata}
Both learned operators, \HMLinc and \HMLfull, take the same input data. They differ, however, in their learning targets: \HMLfull directly predicts the observed radiances \texttt{obs\_rad}, whereas \HMLinc predicts two residual fields, namely \texttt{obs\_rad - rad\_allsky} and \texttt{obs\_rad - rad\_clearsky}. The corresponding input and target data formats are described in detail below.

\subsubsection{Input Data}
From the raw data we create the input data as a tensor of shape \texttt{[num\_data, 175, 214, 14, 68]}  where $175\times 214$ corresponds to the grid of lat/long, and $14$ corresponds to the fourteen input variables depicted in Table~\ref{tbl:inputVariables} and $68$ denotes $68$ profiles corresponding to the levels (height). The data has been carefully chosen by domain-experts as a proper subset of avaialble data to be used as an initial feature candidate set, see Tables~\ref{tbl:data} and~\ref{tbl:inputVariables}. 

\begin{table}[!ht]
	\centering
	\caption{Input Variables}
	\begin{tabular}{ll} 
		\toprule
		\textbf{Data} & \textbf{Unit}\\
		\midrule
		Pressure & \unit{hPa}\\
		Specific humidity & \unit{kg/kg}\\
		Temperature & \unit{K}\\
		Diagnostic cloud water & \unit{g/m^3}\\
		Diagnostic cloud ice & \unit{g/m^3}\\
		Total air density & \unit{kg/m^3}\\
		Cloud fraction & --\\
		Sun zenith angle & \unit{\deg}\\
		Sun azimuth angle & \unit{\deg}\\
		Satellite zenith angle & \unit{\deg}\\
		Satellite azimuth angle & \unit{\deg}\\
		Sinus-encoded date & --\\
		Cosinus-encoded date & --\\
		Daytime & full hours as integer\\
		\bottomrule
	\end{tabular}
	\label{tbl:inputVariables}
\end{table}

\begin{remark}
In order to take into account the seasonality of time, i.e.\ the essential periodicity of time $t$ throughout the year we encode the date $t$ of the observation into a sinus-encoded and cosines encoded date by the formulas\footnote{assuming 365 days per year} $t \to \sin(2\pi t/365)$  and $t \to \cos(2\pi t/365)$ resp.\ resulting in a value range in $[-1,1]$.
\end{remark}

\subsubsection{Output Data}
For results we create output data according to the data depicted in Table~\ref{tbl:outputVariables}. This results in a tensor shape of \texttt{[num\_data, 175, 214, d, 10]}, where $175\times 214$ corresponds to the same grid of lat/long as in the input data, and $10$ to the $10$ satellite channels. The dimension $d\in\{2,1\}$ is different for the two learned operators. As \HMLinc learns the two residuals \texttt{obs\_rad - rad\_allsky} and \texttt{obs\_rad - rad\_clearsky}, we have $d=2$ in this case and as  \HMLfull only learns a single value, \texttt{obs\_rad}, we get $d=1$. 

\begin{table}[!ht]
	\centering
	\caption{Output Variables}
	\begin{tabular}{lll} 
		\toprule
		\textbf{Model} & \textbf{Data} & \textbf{Unit} \\
		\midrule
		\multirow{2}{4em}{\HMLinc} &obs\_rad - rad\_allsky&Percentages for Channels 1-3 or \unit{K} otherwise\\ 
		&obs\_rad - rad\_clearsky&Percentages for Channels 1-3 or \unit{K} otherwise \\ 
		\midrule
		\HMLfull &obs\_rad &Percentages for Channels 1-3 or \unit{K} otherwise\\ 
		\bottomrule
	\end{tabular}
	\label{tbl:outputVariables}
\end{table}

\subsection{Data Flow and Data Cleansing}
\label{data-flow}
The data flow is comprised of the following steps:
\begin{enumerate}
    	\item The raw data is mapped onto the grid as described in ~Subsection~\ref{sec:raw2grid}.
	\item From raw data we create the input tensor of shape \texttt{[num\_data, 175, 214, 14, 68]} as described above.
    Our grid covers central Europe, i.e. \qty{0}{\degree} to \qty{17}{\degree} longitude and \qty{43.5}{\degree} to \qty{57.5}{\degree} latitude with a \qty{0.08}{\degree} resolution. 
	\item Data Cleansing: Due to the nature of the satellite measurement wavelengths, some channels of \texttt{obs\_rad} carried infinite-values during the night, mostly the first three channels (see Table~\ref{tbl:channels}). These infinite-values are interpreted as ``no measurement'' and are stored as fill-values of magnitude $1 \times 10^{36}$. 
    As infinite or very large values lead to exploding gradients and an unstable or even impossible training process, we mask all values exceeding $1 \times 10^{30}$ and map them to NaN-values. This threshold was selected to reliably capture the fill-values of magnitude $1 \times 10^{36}$ used to represent ``no measurement'' cases, while being well below this range to ensure robust detection. The threshold of $1 \times 10^{30}$ is completely outside the valid measurement range, making it safe for identifying problematic values.
	\item After cleansing, the data then is processed through the U-NET architecture as described above.
	\item The output for \HMLinc and \HMLfull respectively is produced for all of the $10$ satellite channels corresponding to different wavelengths as described in Table~\ref{tbl:outputVariables}. 
     \item As a result of data cleansing we have converted large values to NaN, we employ a NaN-aware loss functions (RMSE, MSE) that exclude NaN-values from both predictions and targets during loss computation.
	\item For detailed analysis, the results are mapped back to the raw coordinates as described in Subsection~\ref{sec:grid2raw}.
\end{enumerate}

\subsection{Training}    
We train both models by minimizing MSE between the model predictions and the respective targets, the two residual values for \HMLinc and the single \texttt{obs\_rad} value for \HMLfull. Both models use the Adam optimizer with a one-cycle learning rate policy \cite{smith2019super}, with a warm-up phase spanning \qtyrange{35}{40}{\percent} of total training steps followed by cosine annealing.
    
As mentioned above to evaluate our incremental correction model against RTTOV, we also reconstruct \texttt{obs\_rad} from our learned residuals and compare it to the target \texttt{obs\_rad} from the SEVIRI data. 

In our experiments, the most stable runs consistently used small batch sizes (e.g., 4 and 8) and peak learning rates on the order of $1 \times 10^{-6}$ to $1 \times 10^{-7}$; a systematic hyperparameter search across batch size and learning rate combinations was outside the scope of this work.

\begin{figure}[!ht]
    \centering
     \includegraphics[width=0.6\linewidth]{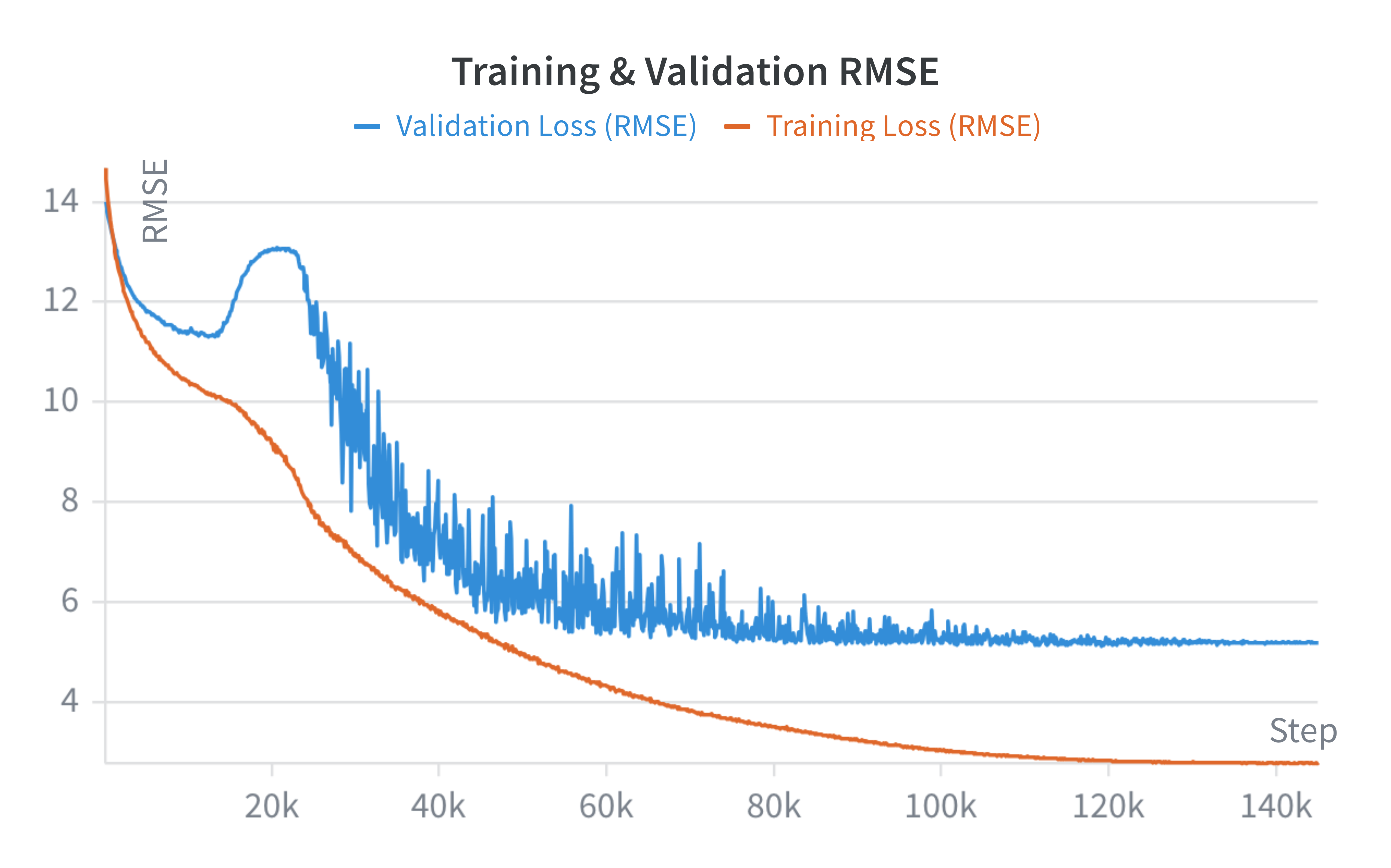}
    \caption{Training (orange) and validation (blue) loss (RMSE)}
    \label{fig:TrainValLossCurve}
\end{figure}

Fig.~\ref{fig:TrainValLossCurve} shows the training and validation loss curves computed on the residual targets, converging to a validation RMSE of approximately $5$. The early fluctuations between steps \num{20000} and \num{40000} are attributable to the learning rate warm-up phase, after which both curves stabilize and converge smoothly.

\section{Results}
In this section we evaluate the two learned operators, the full-surrogate model and the incremental correction operator, against the reference satellite observations. Performance is assessed both visually (sample images and residual maps) and quantitatively (RMSE and MSE) on the held-out test set, analyzing results in both the regular grid domain and the original raw satellite coordinate system.

\subsection{Visual Inspection}
Figures~\ref{fig:compSample15Channel5} and ~\ref{fig:compSample15Channel3} represent a typical example of a visual comparison between RTTOV simulations (left column), satellite observations (top right), our incremental model's predicted reconstructions (middle column) and our fully data-driven model's prediction (bottom right). 
Furthermore, channel-wise comparisons of \HMLinc against RTTOV and satellite observations can be found in the appendix (Figs.~\ref{fig:channelOverview1}--\ref{fig:channelOverview3}).

\begin{figure}[!h]
    \centering
        \subfloat[Comparison of outputs of \HRTTOV (left column), \HMLinc (middle column) and \HMLfull (bottom right) versus the observed radiance \texttt{obs\_rad} (top right) in grid domain. ]
        {
        \label{fig:compSample15Channel5}
        \includegraphics[width=1\linewidth]{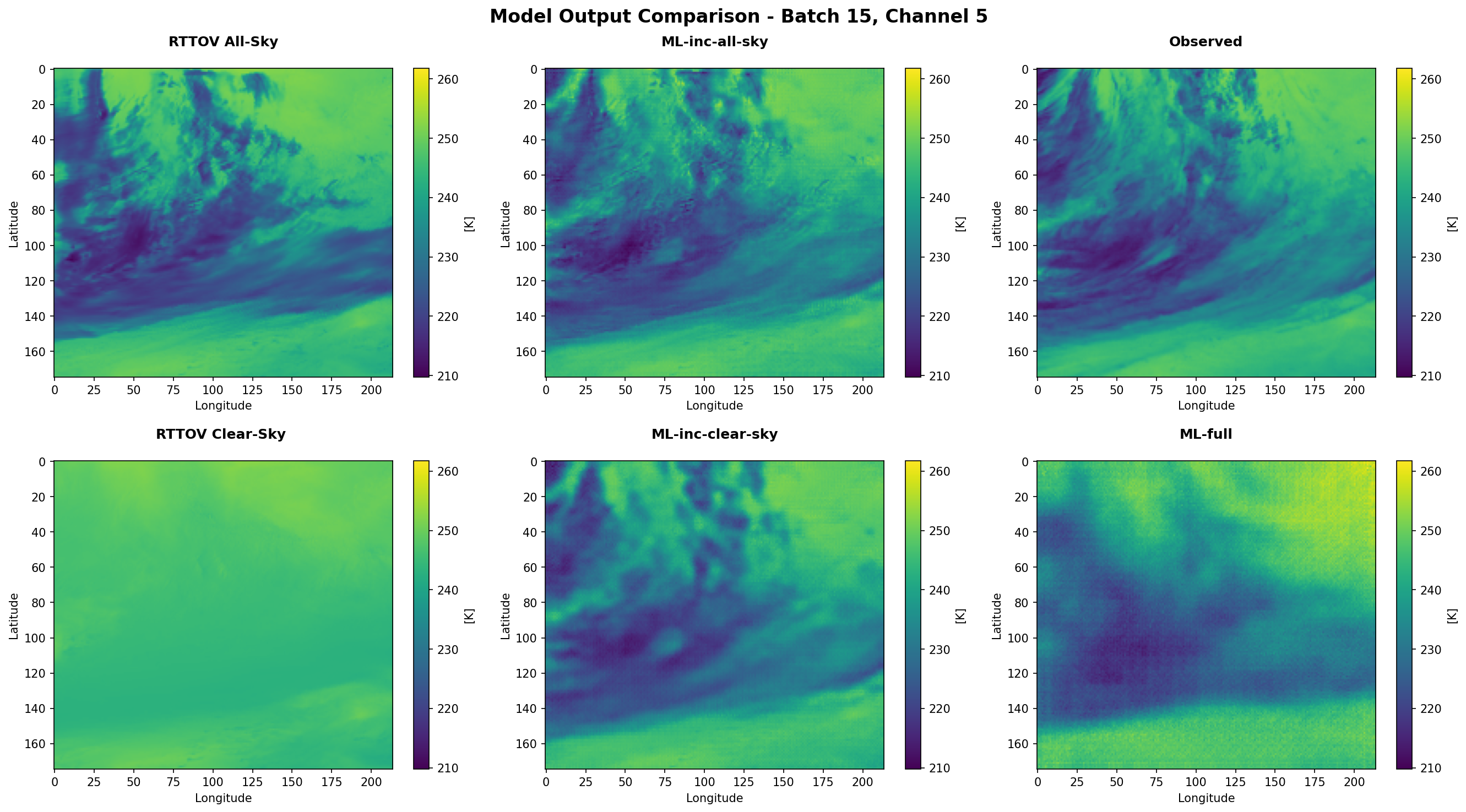}
        }\\
        \subfloat[Difference plots between \HRTTOV (left column), \HMLinc (middle column) and \HMLfull (bottom right) and the observed radiance \texttt{obs\_rad} (top right) in grid domain. ]
        {
        \label{fig:diffSample15Channel5}
        \includegraphics[width=1\linewidth]{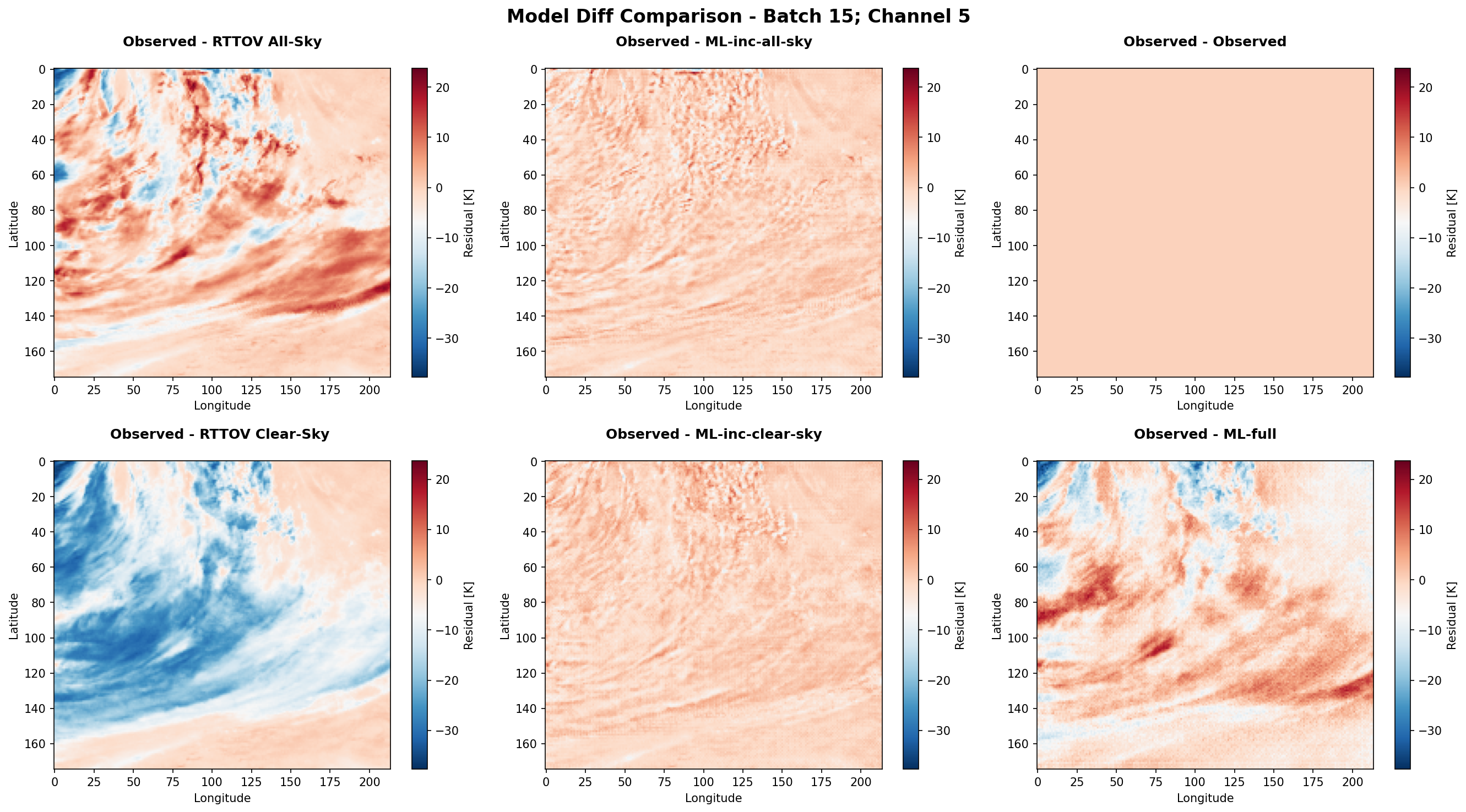}
        }  
    \caption{Model Comparison of \HRTTOV, \HMLinc, and \HMLfull versus satellite observation on sample 15, channel 5 from our testing dataset.}
    \label{fig:sample15Channel5}
\end{figure}
    
\begin{figure}[!h]
    \centering
        \subfloat[Comparison of outputs of \HRTTOV (left column), \HMLinc (middle column) and \HMLfull (bottom right) versus the observed radiance \texttt{obs\_rad} (top right) in grid domain.]
        {
        \label{fig:compSample15Channel3}
        \includegraphics[width=1\linewidth]{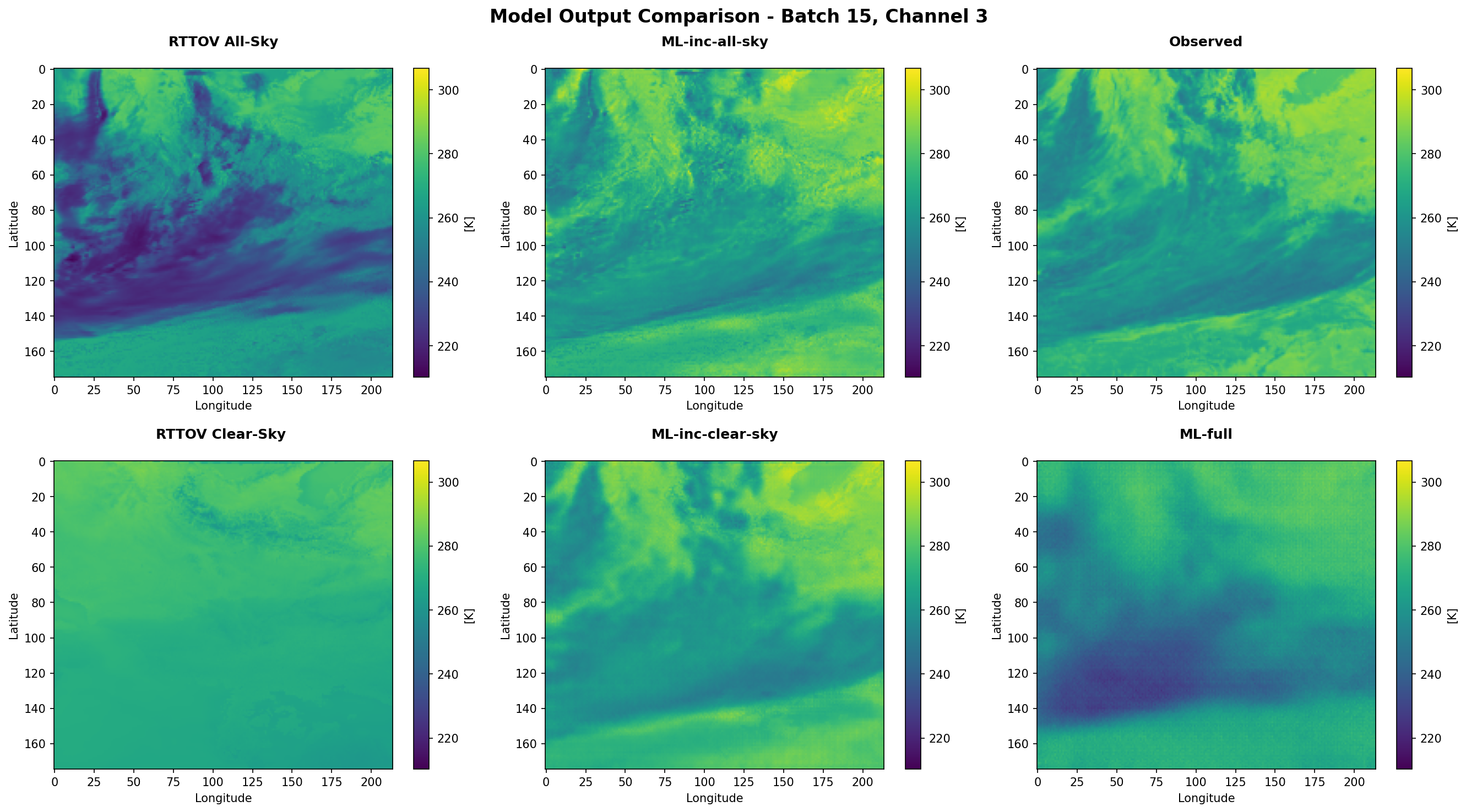}
        }\\
        \subfloat[Difference plots between \HRTTOV (left column), \HMLinc (middle column) and \HMLfull (bottom right) and the observed radiance \texttt{obs\_rad} (top right) in grid domain.]
        {
        \label{fig:diffSample15Channel3}
        \includegraphics[width=1\linewidth]{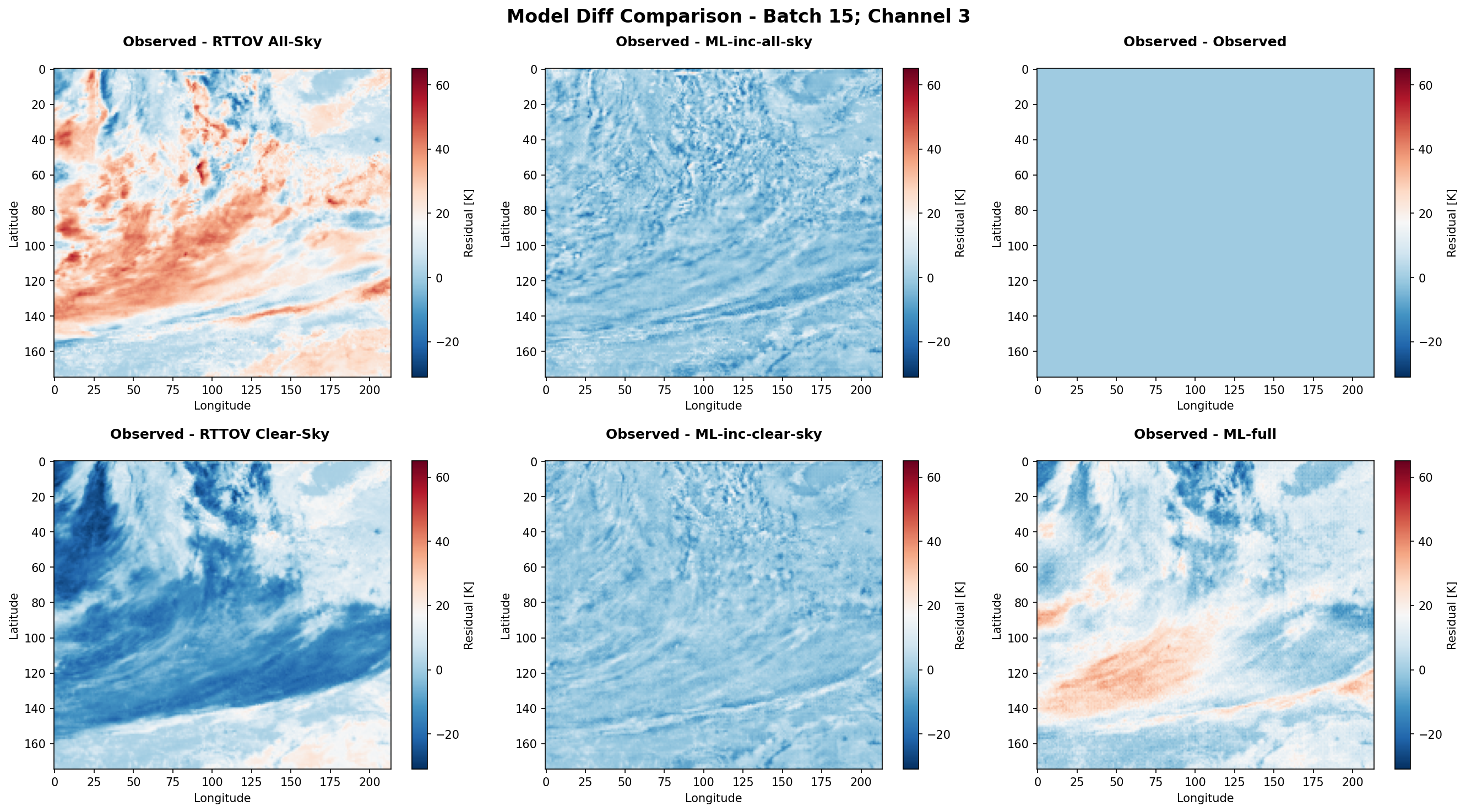}
        }  
    \caption{Model Comparison of \HRTTOV, \HMLinc, and \HMLfull against satellite observation for sample 15, channel 3 from our unseen testing dataset.}
    \label{fig:sample10Channel3}
\end{figure}

Both of our reconstructions demonstrate notably better agreement with observations compared to RTTOV, particularly in preserving fine spatial details. The RTTOV clear-sky image exhibits significantly less sharp details and higher-frequency features than our prediction, a finding that is reflected in the quantitative results from Table~\ref{tbl:lossesOverview}, which shows a substantial difference in RMSE between our approach and the RTTOV model.
Examining the all-sky images in Fig.~\ref{fig:compSample15Channel5} and ~\ref{fig:compSample15Channel3}, more carefully, we observe that while the RTTOV image appears superficially similar to the observation, fine details, particularly in the top-right and left corners, show noticeably larger deviations from the observation compared to our approach. This observation is further supported by the difference plots presented in Fig.~\ref{fig:diffSample15Channel5} and ~\ref{fig:diffSample15Channel3}, which display the pixel-wise differences between the observation (top right) and RTTOV (left column) versus our incremental approach (middle column) and full surrogate (bottom right). The RTTOV difference images exhibit highly structured patterns that closely resemble features from the original images in Fig.~\ref{fig:diffSample15Channel5} and ~\ref{fig:diffSample15Channel3}, suggesting systematic biases in the RTTOV reconstruction. In contrast, the difference images from our approach appear predominantly noise-like with minimal structural patterns. This noise-like characteristic of the residuals indicates that our model successfully captures the systematic atmospheric features, leaving primarily random errors rather than systematic biases. Hence, this result is consistent with our hypothesis from Section~\ref{sec:methodology} that the remaining error in the decomposition in Eq.~\ref{eq:errorDecomp3} largely reflects random sensor measurement noise. A quantitative statistical validation of this hypothesis is left for future work.

\subsection{Quantitative Evaluation}
In the following subsection we quantitatively evaluate and compare the two learned operators and the baseline RTTOV operator by computing the RMSE against the reference satellite observations.

\subsubsection{Training and Validation Losses and Visual Inspection}
To determine if our incremental, hybrid correction operator \HMLinc is able to keep up with the performance of the original RTTOV model \HRTTOV, we compared our learned residuals with the residuals between RTTOV and the observed reflectance values (\texttt{obs\_rad}) using RMSE and MSE.
For training, we used either RMSE or MSE as the loss function, with the remaining metric to monitor training and validation performance. We observed no significant differences in model performance or convergence behavior between the two choices for the loss function.
We evaluated multiple U-Net variants with different architectural complexities, ranging from three to five encoder-decoder blocks with corresponding parameter counts from approximately 40M to 150M. Despite this range in model complexity, we observed no significant performance improvements or enhanced convergence with deeper or wider architectures. Consequently, we selected the 4-block encoder-decoder U-Net configuration with 78M parameters for our final implementation.
Furthermore, for the expanding-path we experimented with using fixed upsampling or transposed convolution. We observed less artifacts using transposed convolutions versus fixed upsampling, both followed by a convolution block.

\subsubsection{Comparison of the fully data-driven surrogate model \HMLfull against numerical RTTOV \HRTTOV}

We quantitatively evaluate the performance of our full surrogate model \HMLfull against the numerical RTTOV model \HRTTOV through comprehensive error metrics, as shown in Table~\ref{tbl:fullSurrogateLossesOverview}. 
Therefore, we calculate the error metrics of the full surrogate against the observed reflectance values as provided by the SEVIRI data directly without any reconstructions that involve RTTOV outputs. 

Table ~\ref{tbl:fullSurrogateLossesOverview} shows the comparison of the predictive performance of \HMLfull versus  \HRTTOV by means of the average RMSE against the observed reflectance target values from the testing dataset. As can be seen, the average RMSE between the full surrogate models predictions to the observed targets is smaller than from the numerical RTTOV model. Note, that the full surrogate does not produce separate predictions for clear-sky and all-sky like the RTTOV. Hence, the average RMSE-values in the table are the same for clear-sky and all-sky for the full surrogate. Our fully data-driven surrogate outperforms the RTTOV model in both domains (raw and grid) and operation modes (clear-sky and all-sky). 

\begin{table}[!ht]
    \centering
    \caption{Mean and standard deviation of RMSE comparing \HMLfull and \HRTTOV.}
    \begin{tabular}{l
		S[separate-uncertainty, uncertainty-separator={}, table-format=2.2(4)]
		S[separate-uncertainty, uncertainty-separator={}, table-format=2.2(4)]
		S[separate-uncertainty, uncertainty-separator={}, table-format=2.2(5)]}
	\toprule
	& {\textbf{Avg. RMSE(\HMLfull)}}
	& {\textbf{Avg. RMSE(\HRTTOV)}}
	& {\textbf{Avg. RMSE(\HMLfull)/RMSE(\HRTTOV)}} \\
	\midrule
	Clearsky (grid) & 8.30  \pm 1.50 & 25.46 \pm 6.52  & 33.68 \pm 6.26\ \%  \\
	Clearsky (raw)  & 8.59  \pm 1.61  & 26.01 \pm 6.73  & 34.08 \pm 6.18\ \%  \\
	Allsky (grid)   & 8.30  \pm 1.50  & 8.82  \pm 0.96  & 94.23 \pm 13.91\ \% \\
	Allsky (raw)    & 8.59  \pm 1.61  & 9.27  \pm 1.01  & 92.66 \pm 13.04\ \% \\
	\bottomrule
\end{tabular}
    \label{tbl:fullSurrogateLossesOverview}
\end{table}

\subsubsection{Comparison of the hybrid correction operator \HMLinc against numerical RTTOV \HRTTOV}
We quantitatively evaluate the performance of our machine learning operator \HMLinc against the numerical RTTOV operator \HRTTOV through comprehensive error metrics presented in Fig.~\ref{fig:rttovVsMl} and Table~\ref{tbl:lossesOverview}. Detailed sample-wise loss values are provided in the appendix (Tables~\ref{tbl:allskyLossesRaw}--\ref{tbl:clearskyLossesGrid}).

Fig.~\ref{fig:rttovVsMl} (top row) displays the RMSE comparison between both models across our test dataset in the grid domain. The smoothed loss curves (blue and orange dashed lines) and overall mean RMSE (red dashed line) demonstrate that \HMLinc consistently achieves lower reconstruction errors than \HRTTOV when predicting satellite observations (\texttt{obs\_rad}). This performance advantage is quantified in Table
\ref{tbl:lossesOverview}, which shows substantial improvements across all evaluated metrics.
To ensure a comprehensive evaluation, we implemented a reverse-mapping method (described in Subsection~\ref{sec:grid2raw}) that transforms our gridded predictions back to the original raw satellite data format through interpolation. This step is crucial because our gridding process, which creates tensors of shape \texttt{[batch\_size, 175, 214, 2, 10]}, retains only approximately \qty{40}{\percent} of the original satellite observations (\num{37450} gridded points versus \num{96326} raw observations). This reverse-mapping enables direct comparison between our data-driven approach and RTTOV in the native satellite data space.

Table~\ref{tbl:lossesOverview} and the bottom row of Fig.~\ref{fig:rttovVsMl}  present the comparison in the raw data domain, confirming that \HMLinc significantly outperforms \HRTTOV in both grid and raw representations.

\begin{table}[!ht]
    \centering
    \caption{Mean and standard deviation of RMSE comparing the residual-learning model \HMLinc and \HRTTOV.}
    \begin{tabular}{l
		S[separate-uncertainty, uncertainty-separator={}, table-format=2.2(4)]
		S[separate-uncertainty, uncertainty-separator={}, table-format=2.2(4)]
		S[separate-uncertainty, uncertainty-separator={}, table-format=2.2(5)]}
	\toprule
	& {\textbf{Avg. RMSE(\HMLinc)}}
	& {\textbf{Avg. RMSE(\HRTTOV)}}
	& {\textbf{Avg. RMSE(\HMLinc)/RMSE(\HRTTOV)}} \\
	\midrule
	Clearsky (grid) & 3.70 \pm 1.38  & 25.46 \pm 6.52  & 15.31 \pm 6.89\ \%  \\
	Clearsky (raw)  & 3.99 \pm 1.37  & 26.01 \pm 6.73  & 16.19 \pm 6.87\ \%  \\
	Allsky (grid)   & 4.52 \pm 1.30  & 8.82  \pm 0.96  & 51.41 \pm 14.96\ \% \\
	Allsky (raw)    & 4.99 \pm 1.28  & 9.27  \pm 1.01  & 54.09 \pm 13.97\ \% \\
	\bottomrule
\end{tabular}
    \label{tbl:lossesOverview}
\end{table}

\begin{figure}[!ht]
    \centering
    \includegraphics[width=\linewidth]{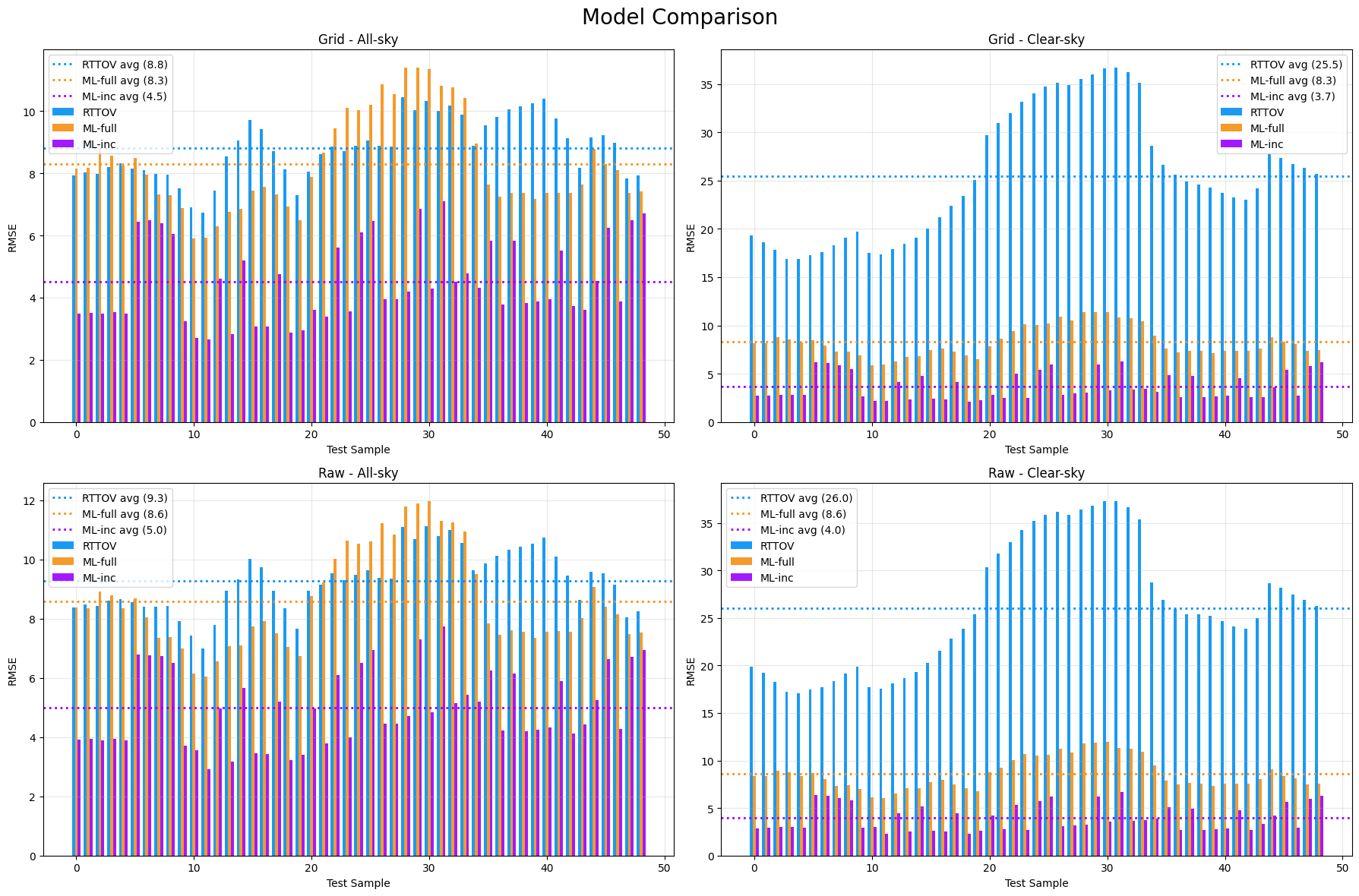}
    \caption{Comparison of the ML models \HMLfull and \HMLinc against \HRTTOV for estimating observed reflectance (\texttt{obs\_rad}) in grid and raw data domains.}		
    \label{fig:rttovVsMl}
\end{figure}

Finally, Fig.~\ref{fig:sample10Channel3} clearly underlines the superiority of the incremental surrogate model compared to RTTOV and the full surrogate. The difference images for the incremental approach (middle column) exhibit the lowest deviation values, with errors that lack the structured atmospheric patterns evident in the RTTOV and full emulator predictions. This indicates that the incremental model successfully captures the systematic biases present in the physical radiative transfer model while maintaining predictive accuracy across both clear-sky and all-sky conditions. In contrast, the full surrogate, despite being trained on the same target observations, yet without residuals, retains significant structured errors comparable to those of RTTOV itself, suggesting that direct learning of the full radiative transfer mapping is more challenging than learning the residual corrections. These results support our hypothesis that learning systematic errors, from which we can reconstruct improved predictions using existing RTTOV outputs, is more effective than directly learning target observations, as described in Subsection~\ref{subsec:datadrivenObsOpSurrogate}.

\subsection{Computational Requirements and Performance Constraints}
While the training process of neural networks generally is resource and time consuming, they can be noteworthy fast during inference. 
Our 78M parameter model was trained on a single Nvidia A100 GPU with \qty{80}{\giga \byte} of GPU memory, which was more than enough considering, that we mostly trained with smaller batch sizes of $4$ or $8$. The training time highly depends on the batch size and epochs. From our experience we achieved the best results with small batch sizes, small learning rates and long training times of around $1000$ epochs. Therefore, the average training run took us around \qtyrange[range-units = single]{15}{25}{hours}.
However, once trained the U-Net can be easily inferenced on a standard consumer GPU. We tested the prediction time per sample on a normal GPU-workstation with $16$ CPU-cores, \qty{64}{\giga \byte} DDR5-RAM and an Nvidia RTX 2080 Ti. Each prediction resembles the output of the RTTOV model for \qty{1}{\hour}. The average prediction time of our 50 test-samples was \qty{65}{\milli\second}, requiring only around \qty{3.5}{\giga \byte} of GPU memory. Predicting and reconstructing \texttt{obs\_rad} from the predicted residuals and RTTOV outputs therefore introduces no significant computational or temporal overhead.

\section{Summary and Future Work}
\label{sec:futureWork}
We developed and evaluated two machine-learning surrogate observation operators for assimilating SEVIRI satellite radiances in NWP: a fully data-driven emulator \HMLfull that maps ICON state variables directly to observed radiances, and a hybrid incremental correction model \HMLinc that learns state-dependent residuals relative to RTTOV while retaining the established physics of the baseline forward model. Both approaches exploit the 3D structure of the problem by encoding 14 atmospheric variables over 68 vertical levels with 3D convolutions and then learning horizontal spatial correlations and inter-channel dependencies on a regularized lat–lon grid using a 2D U-Net, with an additional Grid2Raw remapping to assess performance in the original observation geometry. On three weeks of Meteosat-10 data over central Europe (10 channels), the residual-learning formulation is theoretically justified under squared loss as targeting the conditional bias of RTTOV and empirically yields the most accurate reconstructions: \HMLinc produces lower RMSE than RTTOV (clear-sky and all-sky) and also outperforms the direct emulator \HMLfull, with residual error patterns that are largely noise-like rather than structurally biased. Inference is fast (order of \qtyrange{10}{100}{\milli\second} per \qty{1}{\hour} sample on a single consumer GPU), indicating that the hybrid surrogate can improve radiance accuracy with only modest computational overhead, making it a promising candidate for integration into operational data assimilation workflows.

For the current proof-of-concept research study, the dataset covers central Europe over \qty{500}{\hour}-period for approximately three weeks. Further work is needed towards a consolidated and ready-for-operations system. To comprehensively evaluate and enhance the performance of our proposed model, training and testing will be conducted on more comprehensive data, e.g., spanning one complete annual cycle and global coverage. This expanded scope will necessitate either scaling up the model architecture or implementing a mixture-of-experts (MoE) approach, where specialized models are trained for different seasons, geographic regions or day-/night-time.
In addition, we will integrate the operators \HMLfull or \HMLinc into data assimilation and evaluate its performance in the complete data assimilation cycle.

\section{Acknowledgements}
The authors would like to thank Matthias Mages for preparing and providing the
training data set.

The work of Gian Luca Buono was partially supported by the ‘‘Mittelbau’’ Funding Program of the Frankfurt University of Applied Sciences.

The authors gratefully acknowledge computing infrastructure funded by the German Federal Ministry of Education and Research (BMBF) under the program ``Forschung an Fachhochschulen, KI-NACHWUCHS@FH 1--2021'', FKZ \mbox{13FH027KI1}.

\newpage
\section{Appendix}

\subsection{Detailed Test Results}

\begin{table}[!ht]
	\centering
	\begin{minipage}[t]{0.48\textwidth}
		\centering
		\scalebox{0.7}{
			\begin{tabular}{l 
S[table-format=2.2, round-mode=places, round-precision=2] 
S[table-format=2.2, round-mode=places, round-precision=2] 
S[table-format=2.2, round-mode=places, round-precision=2]}
\toprule
 & \textbf{{RMSE(ML)}} & \textbf{{RMSE(RTTOV)}} & \textbf{{ML/RTTOV [\%]}} \\
\midrule
0 & 3.921348 & 8.368623 & 46.857741 \\
1 & 3.934717 & 8.482639 & 46.385524 \\
2 & 3.881948 & 8.439108 & 45.999501 \\
3 & 3.942226 & 8.606335 & 45.806099 \\
4 & 3.882009 & 8.655313 & 44.851163 \\
5 & 6.780242 & 8.550002 & 79.301062 \\
6 & 6.765234 & 8.417478 & 80.371274 \\
7 & 6.751748 & 8.399593 & 80.381845 \\
8 & 6.508880 & 8.422389 & 77.280682 \\
9 & 3.712123 & 7.915528 & 46.896723 \\
10 & 3.563308 & 7.442647 & 47.876885 \\
11 & 2.920242 & 6.991107 & 41.770810 \\
12 & 4.980728 & 7.792903 & 63.913639 \\
13 & 3.184726 & 8.945734 & 35.600500 \\
14 & 5.664086 & 9.338284 & 60.654458 \\
15 & 3.467156 & 10.025355 & 34.583873 \\
16 & 3.432566 & 9.752615 & 35.196370 \\
17 & 5.192183 & 8.956417 & 57.971658 \\
18 & 3.231350 & 8.360008 & 38.652470 \\
19 & 3.415752 & 7.666590 & 44.553727 \\
20 & 4.976336 & 8.948264 & 55.612306 \\
21 & 3.800889 & 9.154444 & 41.519609 \\
22 & 6.101266 & 9.528347 & 64.032789 \\
23 & 4.001647 & 9.304893 & 43.005838 \\
24 & 6.516036 & 9.472341 & 68.790132 \\
25 & 6.937171 & 9.642916 & 71.940600 \\
26 & 4.450380 & 9.393949 & 47.374970 \\
27 & 4.459311 & 9.359281 & 47.645869 \\
28 & 4.707335 & 11.111302 & 42.365281 \\
29 & 7.304235 & 10.685128 & 68.358885 \\
30 & 4.846346 & 11.133723 & 43.528533 \\
31 & 7.729683 & 10.784902 & 71.671339 \\
32 & 5.138890 & 11.001439 & 46.711073 \\
33 & 5.421949 & 10.566756 & 51.311384 \\
34 & 5.210536 & 9.631210 & 54.100527 \\
35 & 6.261888 & 9.861984 & 63.495209 \\
36 & 4.216835 & 10.115948 & 41.685022 \\
37 & 6.157926 & 10.340358 & 59.552346 \\
38 & 4.198570 & 10.431516 & 40.248895 \\
39 & 4.250365 & 10.544579 & 40.308532 \\
40 & 4.318197 & 10.738234 & 40.213288 \\
41 & 5.884327 & 10.111309 & 58.195501 \\
42 & 4.125679 & 9.453641 & 43.641170 \\
43 & 4.429915 & 8.641105 & 51.265615 \\
44 & 5.247677 & 9.583621 & 54.756723 \\
45 & 6.647925 & 9.529996 & 69.757899 \\
46 & 4.277202 & 9.152203 & 46.734129 \\
47 & 6.715635 & 8.037416 & 83.554661 \\
48 & 6.946316 & 8.262883 & 84.066495 \\
\toprule
\textbf{AVG} & \textbf{4.99} & \textbf{9.27} & \textbf{54.09} \\
\textbf{STD} & \textbf{1.28} & \textbf{1.012} & \textbf{13.97} \\
\bottomrule
\end{tabular}

		}
		\vspace{0.5em}
		\caption{Mean absolute errors for RTTOV \texttt{all-sky} and our reconstructed observed radiances versus ground truth targets (\texttt{obs\_rad}) in raw data format.}
		\label{tbl:allskyLossesRaw}
	\end{minipage}
	\hfill
	\begin{minipage}[t]{0.48\textwidth}
		\centering
		\scalebox{0.7}{
			\begin{tabular}{l 
S[table-format=2.2, round-mode=places, round-precision=2] 
S[table-format=2.2, round-mode=places, round-precision=2] 
S[table-format=2.2, round-mode=places, round-precision=2]}
\toprule
 & \textbf{{RMSE(ML)}} & \textbf{{RMSE(RTTOV)}} & \textbf{{ML/RTTOV [\%]}} \\
\midrule
0 & 2.884744 & 19.884565 & 14.507455 \\
1 & 2.970073 & 19.215059 & 15.457008 \\
2 & 3.004642 & 18.293678 & 16.424480 \\
3 & 3.009206 & 17.195377 & 17.500085 \\
4 & 2.968961 & 17.101734 & 17.360582 \\
5 & 6.336011 & 17.474297 & 36.259036 \\
6 & 6.256692 & 17.724876 & 35.298931 \\
7 & 6.074660 & 18.344368 & 33.114577 \\
8 & 5.798849 & 19.176458 & 30.239416 \\
9 & 2.934908 & 19.865141 & 14.774163 \\
10 & 3.014849 & 17.731361 & 17.002920 \\
11 & 2.295997 & 17.516531 & 13.107603 \\
12 & 4.428965 & 18.090887 & 24.481744 \\
13 & 2.545656 & 18.655031 & 13.645951 \\
14 & 5.152756 & 19.279621 & 26.726437 \\
15 & 2.611323 & 20.241711 & 12.900701 \\
16 & 2.569005 & 21.573692 & 11.908043 \\
17 & 4.418277 & 22.864285 & 19.323925 \\
18 & 2.290700 & 23.898760 & 9.585018 \\
19 & 2.593123 & 25.392820 & 10.212032 \\
20 & 4.247819 & 30.330402 & 14.005152 \\
21 & 2.742141 & 31.737505 & 8.640064 \\
22 & 5.335839 & 32.966290 & 16.185744 \\
23 & 2.709124 & 34.258705 & 7.907842 \\
24 & 5.706462 & 35.173702 & 16.223661 \\
25 & 6.225390 & 35.833237 & 17.373229 \\
26 & 3.077540 & 36.164845 & 8.509755 \\
27 & 3.159155 & 35.869785 & 8.807287 \\
28 & 3.281357 & 36.417767 & 9.010318 \\
29 & 6.208377 & 36.785130 & 16.877409 \\
30 & 3.573019 & 37.309540 & 9.576691 \\
31 & 6.671118 & 37.293736 & 17.888040 \\
32 & 3.655398 & 36.677650 & 9.966281 \\
33 & 3.762561 & 35.376389 & 10.635798 \\
34 & 3.877836 & 28.755007 & 13.485776 \\
35 & 5.093611 & 26.861177 & 18.962725 \\
36 & 2.717621 & 25.927284 & 10.481702 \\
37 & 4.950898 & 25.420675 & 19.475870 \\
38 & 2.674927 & 25.399300 & 10.531500 \\
39 & 2.795217 & 25.213139 & 11.086351 \\
40 & 2.880492 & 24.662951 & 11.679432 \\
41 & 4.771920 & 24.105291 & 19.796152 \\
42 & 2.702681 & 23.871284 & 11.321892 \\
43 & 3.352132 & 24.984854 & 13.416655 \\
44 & 4.202174 & 28.618767 & 14.683282 \\
45 & 5.675609 & 28.195572 & 20.129434 \\
46 & 2.906009 & 27.451435 & 10.586001 \\
47 & 5.972670 & 26.925385 & 22.182302 \\
48 & 6.291343 & 26.227827 & 23.987283 \\
\toprule
\textbf{AVG} & \textbf{3.99} & \textbf{26.01} & \textbf{16.19} \\
\textbf{STD} & \textbf{1.37} & \textbf{6.73} & \textbf{6.87} \\
\bottomrule
\end{tabular}

		}
		\vspace{0.5em}
		\caption{Mean absolute errors for RTTOV \texttt{clear-sky} and our reconstructed observed radiances versus ground truth targets (\texttt{obs\_rad}) in raw data format.}
		\label{tbl:clearskyLossesRaw}
	\end{minipage}
\end{table}

\newpage
\begin{table}[!ht]
	\centering
	\begin{minipage}[t]{0.48\textwidth}
		\centering
		\scalebox{0.8}{
			\begin{tabular}{l 
S[table-format=2.2, round-mode=places, round-precision=2] 
S[table-format=2.2, round-mode=places, round-precision=2] 
S[table-format=2.2, round-mode=places, round-precision=2]}
\toprule
 & \textbf{{RMSE(ML)}} & \textbf{{RMSE(RTTOV)}} & \textbf{{ML/RTTOV [\%]}} \\
\midrule
0 & 3.489170 & 7.921046 & 44.049353 \\
1 & 3.503776 & 8.020137 & 43.687229 \\
2 & 3.483400 & 7.992135 & 43.585344 \\
3 & 3.524593 & 8.192903 & 43.020071 \\
4 & 3.474501 & 8.333333 & 41.694007 \\
5 & 6.442409 & 8.164782 & 78.904850 \\
6 & 6.487770 & 8.109888 & 79.998267 \\
7 & 6.398919 & 7.978305 & 80.203995 \\
8 & 6.053373 & 7.961966 & 76.028622 \\
9 & 3.240528 & 7.522496 & 43.077834 \\
10 & 2.712847 & 6.899773 & 39.317924 \\
11 & 2.661493 & 6.732858 & 39.529920 \\
12 & 4.604964 & 7.440346 & 61.891799 \\
13 & 2.839228 & 8.547999 & 33.215121 \\
14 & 5.189870 & 9.068316 & 57.230805 \\
15 & 3.062625 & 9.716460 & 31.519965 \\
16 & 3.066852 & 9.413484 & 32.579353 \\
17 & 4.754408 & 8.713689 & 54.562515 \\
18 & 2.875252 & 8.125205 & 35.386820 \\
19 & 2.947012 & 7.309178 & 40.319337 \\
20 & 3.599809 & 8.049144 & 44.722876 \\
21 & 3.377066 & 8.608930 & 39.227479 \\
22 & 5.621278 & 8.871445 & 63.363725 \\
23 & 3.550762 & 8.704511 & 40.792207 \\
24 & 6.107119 & 8.888284 & 68.709762 \\
25 & 6.476087 & 9.055721 & 71.513767 \\
26 & 3.942340 & 8.877283 & 44.409305 \\
27 & 3.948040 & 8.852448 & 44.598279 \\
28 & 4.200110 & 10.448020 & 40.200054 \\
29 & 6.848497 & 10.039160 & 68.217826 \\
30 & 4.302286 & 10.322312 & 41.679482 \\
31 & 7.109958 & 10.015122 & 70.992226 \\
32 & 4.523867 & 10.178069 & 44.447202 \\
33 & 4.773357 & 9.881929 & 48.303898 \\
34 & 4.310274 & 8.887510 & 48.498110 \\
35 & 5.829492 & 9.536246 & 61.129844 \\
36 & 3.775347 & 9.816466 & 38.459327 \\
37 & 5.820423 & 10.055048 & 57.885581 \\
38 & 3.837476 & 10.160182 & 37.769754 \\
39 & 3.883986 & 10.257697 & 37.864116 \\
40 & 3.949937 & 10.401467 & 37.974808 \\
41 & 5.505652 & 9.775918 & 56.318513 \\
42 & 3.723112 & 9.140941 & 40.730071 \\
43 & 3.613492 & 8.179025 & 44.179980 \\
44 & 4.536908 & 9.155711 & 49.552768 \\
45 & 6.237505 & 9.233944 & 67.549740 \\
46 & 3.866080 & 8.984185 & 43.032064 \\
47 & 6.484752 & 7.842349 & 82.688898 \\
48 & 6.704850 & 7.934988 & 84.497298 \\
\toprule
\textbf{AVG} & \textbf{4.52} & \textbf{8.82} & \textbf{51.41}\\
\textbf{STD} & \textbf{1.30} & \textbf{0.96} & \textbf{14.96} \\

\bottomrule
\end{tabular}

		}
		\vspace{0.5em}
		\caption{Mean absolute errors for RTTOV \texttt{all-sky} and our reconstructed observed radiances versus ground truth targets (\texttt{obs\_rad}) in grid format.}
		\label{tbl:allskyLossesGrid}
	\end{minipage}
	\hfill
	\begin{minipage}[t]{0.48\textwidth}
		\centering
		\scalebox{0.8}{
			\begin{tabular}{l 
S[table-format=2.2, round-mode=places, round-precision=2] 
S[table-format=2.2, round-mode=places, round-precision=2] 
S[table-format=2.2, round-mode=places, round-precision=2]}
\toprule
 & \textbf{{RMSE(ML)}} & \textbf{{RMSE(RTTOV)}} & \textbf{{ML/RTTOV [\%]}} \\
\midrule
0 & 2.697225 & 19.294400 & 13.979315 \\
1 & 2.769185 & 18.637083 & 14.858467 \\
2 & 2.811202 & 17.800218 & 15.793074 \\
3 & 2.824024 & 16.875980 & 16.733985 \\
4 & 2.779194 & 16.874862 & 16.469429 \\
5 & 6.158066 & 17.283352 & 35.630044 \\
6 & 6.129459 & 17.624349 & 34.778358 \\
7 & 5.874402 & 18.323112 & 32.060068 \\
8 & 5.510039 & 19.095088 & 28.855792 \\
9 & 2.620960 & 19.682953 & 13.315890 \\
10 & 2.154471 & 17.496436 & 12.313772 \\
11 & 2.148781 & 17.335209 & 12.395471 \\
12 & 4.138226 & 17.922139 & 23.090024 \\
13 & 2.317481 & 18.474856 & 12.543975 \\
14 & 4.758239 & 19.086969 & 24.929256 \\
15 & 2.395746 & 20.011839 & 11.971644 \\
16 & 2.374350 & 21.233958 & 11.181854 \\
17 & 4.130107 & 22.414391 & 18.426139 \\
18 & 2.131734 & 23.408196 & 9.106784 \\
19 & 2.253464 & 25.031012 & 9.002689 \\
20 & 2.814201 & 29.696391 & 9.476576 \\
21 & 2.504879 & 30.940991 & 8.095666 \\
22 & 4.984340 & 31.995500 & 15.578252 \\
23 & 2.496278 & 33.148007 & 7.530703 \\
24 & 5.440426 & 34.010569 & 15.996280 \\
25 & 5.946480 & 34.702860 & 17.135419 \\
26 & 2.847459 & 35.132609 & 8.104889 \\
27 & 2.931174 & 34.894820 & 8.400027 \\
28 & 3.017256 & 35.493290 & 8.500919 \\
29 & 5.931662 & 35.977147 & 16.487305 \\
30 & 3.299906 & 36.623004 & 9.010475 \\
31 & 6.269154 & 36.723762 & 17.071110 \\
32 & 3.353099 & 36.255512 & 9.248521 \\
33 & 3.428089 & 35.125661 & 9.759500 \\
34 & 3.114027 & 28.603351 & 10.886931 \\
35 & 4.892410 & 26.642409 & 18.363242 \\
36 & 2.606369 & 25.586525 & 10.186492 \\
37 & 4.805959 & 24.895064 & 19.304867 \\
38 & 2.587599 & 24.623717 & 10.508562 \\
39 & 2.685709 & 24.300006 & 11.052299 \\
40 & 2.752954 & 23.746603 & 11.593042 \\
41 & 4.556905 & 23.232404 & 19.614434 \\
42 & 2.546516 & 23.045552 & 11.049925 \\
43 & 2.573136 & 24.197170 & 10.634038 \\
44 & 3.613062 & 27.767432 & 13.011868 \\
45 & 5.413266 & 27.349832 & 19.792685 \\
46 & 2.739482 & 26.712840 & 10.255299 \\
47 & 5.831012 & 26.287208 & 22.181937 \\
48 & 6.152999 & 25.702339 & 23.939451 \\
\toprule
\textbf{AVG} & \textbf{3.70} & \textbf{25.46} & \textbf{15.31} \\
\textbf{STD} & \textbf{1.38} & \textbf{6.52} & \textbf{6.89} \\
\bottomrule
\end{tabular}

		}
		\vspace{0.5em}
		\caption{Mean absolute errors for RTTOV \texttt{clear-sky} and our reconstructed observed radiances versus ground truth targets (\texttt{obs\_rad}) in grid format.}
		\label{tbl:clearskyLossesGrid}
	\end{minipage}
\end{table}

\newpage
\subsection{Channel-wise output comparison}
\begin{figure}[!ht]
	\centering
	\includegraphics[height=0.91\textheight]{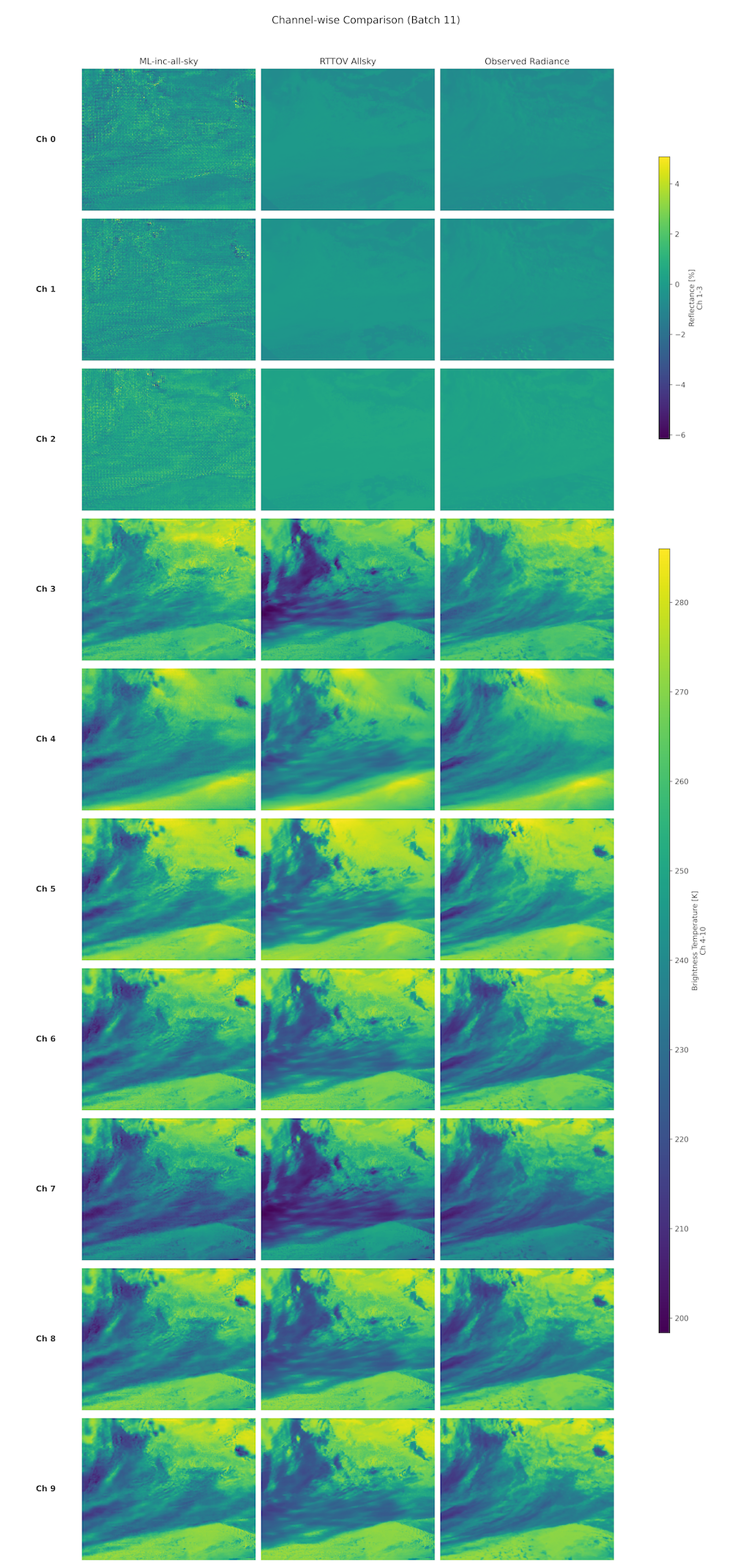}
	\caption{Channel-wise comparison of \HMLinc (left), \HRTTOV (middle), and satellite observations (right) for all 10 channels on test sample 11.}
	\label{fig:channelOverview1}
\end{figure}

\begin{figure}[!ht]
	\centering
	\includegraphics[height=0.91\textheight]{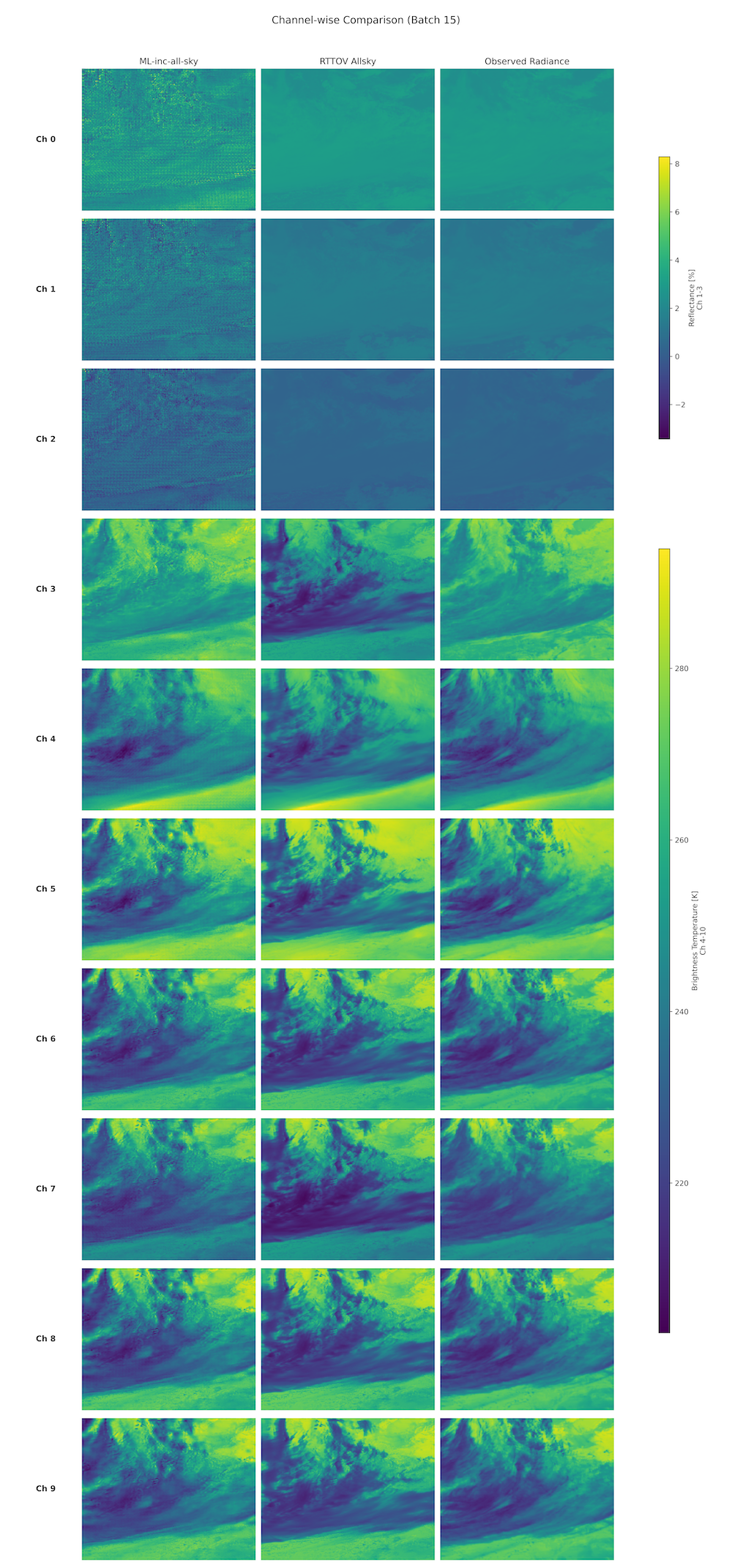}
	\caption{Channel-wise comparison of \HMLinc (left), \HRTTOV (middle), and satellite observations (right) for all 10 channels on test sample 15.}
	\label{fig:channelOverview2}
\end{figure}

\begin{figure}[!ht]
	\centering
	\includegraphics[height=0.91\textheight]{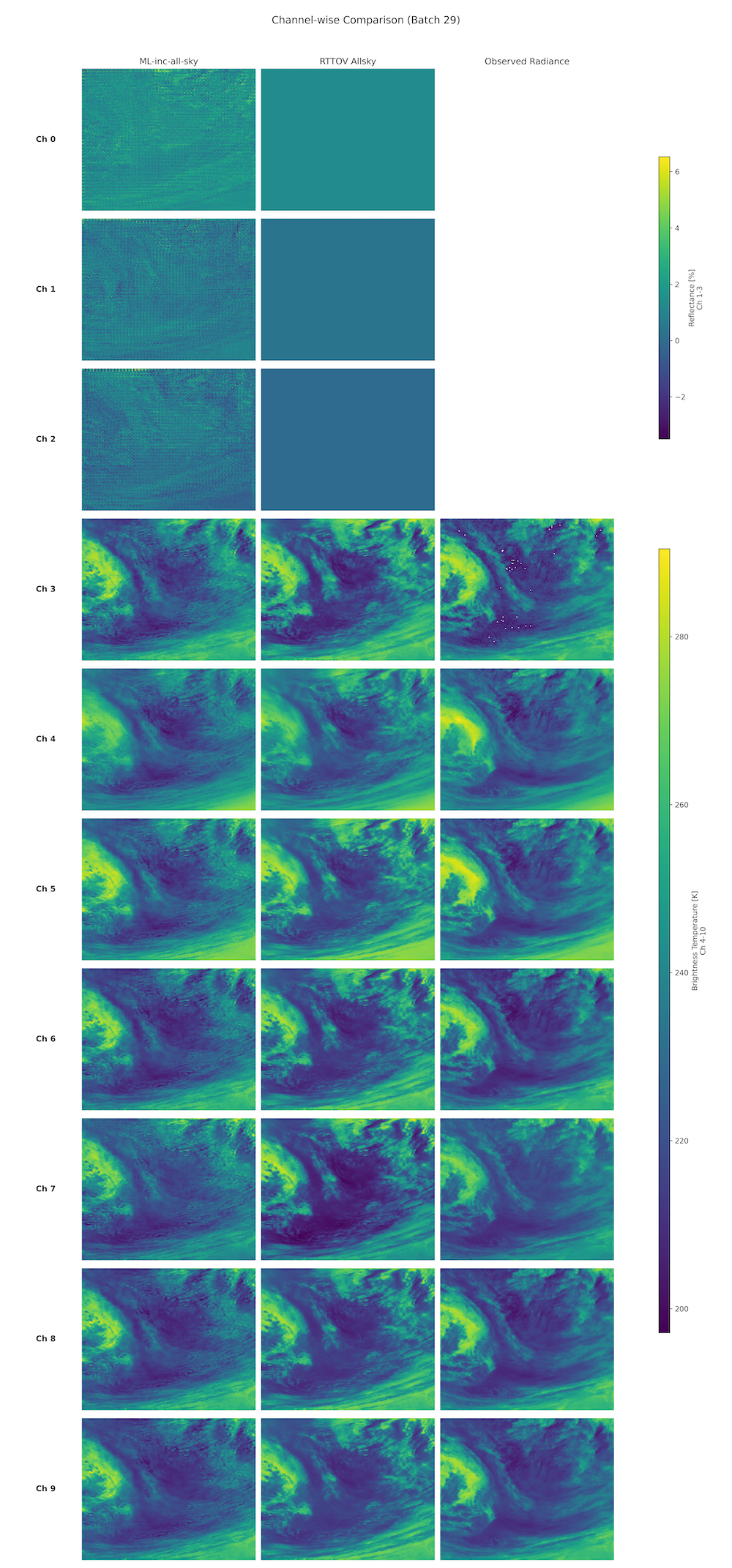}
	\caption{Channel-wise comparison of \HMLinc (left), \HRTTOV (middle), and satellite observations (right) for all 10 channels on test sample 29.}
	\label{fig:channelOverview3}
\end{figure}

\clearpage
\bibliographystyle{alpha}
\bibliography{References}

@article{bauer2015quiet,
	author = {Bauer, Peter and Thorpe, Alan and Brunet, Gilbert},
	doi = {10.1038/nature14956},
	journal = {Nature},
	pages = {47--55},
	title = {The Quiet Revolution of Numerical Weather Prediction},
	url = {https://doi.org/10.1038/nature14956},
	volume = {525},
	year = {2015}}

@book{rodgers2000inverse,
	author = {Rodgers, Clive D.},
	doi = {10.1142/3171},
	publisher = {World Scientific},
	title = {Inverse Methods for Atmospheric Sounding},
	year = {2000}}

@article{saunders2018update,
	author = {Saunders, Roger and Hocking, James and Turner, Emma and Rayer, Peter and Rundle, David and Brunel, Pascal and Vidot, Jerome and Roquet, Pascale and Matricardi, Marco and Geer, Alan and Bormann, Niels and Lupu, Cristina},
	doi = {10.5194/gmd-11-2717-2018},
	journal = {Geoscientific Model Development},
	number = {7},
	pages = {2717--2737},
	publisher = {Copernicus Publications G{\"o}ttingen, Germany},
	title = {An update on the RTTOV fast radiative transfer model (currently at version~12)},
	url = {https://gmd.copernicus.org/articles/11/2717/2018/},
	volume = {11},
	year = {2018}}

@article{bauer2010allsky,
	author = {Bauer, Peter and Geer, Alan and Lopez, Philippe},
	doi = {10.1002/qj.195},
	journal = {Quarterly Journal of the Royal Meteorological Society},
	number = {649},
	pages = {124--139},
	title = {All-sky microwave assimilation and its impact on operational numerical weather prediction},
	volume = {136},
	year = {2010}}

@article{geer2018all,
	author = {Geer, Alan J. and Lonitz, Katrin and Weston, Peter and Kazumori, Masahiro and Okamoto, Kozo and Zhu, Yanqiu and Liu, Emily Huichun and Collard, Andrew and Bell, William and Migliorini, Stefano and Chambon, Philippe and Fourri{\'e}, Nadia and Kim, Min-Jeong and K{\"o}pken-Watts, Christina and Schraff, Christoph},
	doi = {10.1002/qj.3202},
	eprint = {https://rmets.onlinelibrary.wiley.com/doi/pdf/10.1002/qj.3202},
	journal = {Quarterly Journal of the Royal Meteorological Society},
	number = {713},
	pages = {1191-1217},
	title = {All-sky satellite data assimilation at operational weather forecasting centres},
	url = {https://rmets.onlinelibrary.wiley.com/doi/abs/10.1002/qj.3202},
	volume = {144},
	year = {2018}}

@article{krasnopolsky2010nnradiation,
	author = {Krasnopolsky, Vladimir M. and Fox-Rabinovitz, Michael S. and Hou, Y. T. and Lord, S. J. and Belochitski, A. A.},
	doi = {10.1175/2009MWR3149.1},
	journal = {Monthly Weather Review},
	number = {5},
	pages = {1822--1842},
	title = {Accurate and Fast Neural Network Emulations of Model Radiation for the NCEP Coupled Climate Forecast System: Climate Simulations and Seasonal Predictions},
	volume = {138},
	year = {2010}}

@article{dueben2018challenges,
	author = {Dueben, Peter D and Bauer, Peter},
	doi = {10.5194/gmd-11-3999-2018},
	journal = {Geoscientific Model Development},
	number = {10},
	pages = {3999--4009},
	publisher = {Copernicus Publications G{\"o}ttingen, Germany},
	title = {Challenges and design choices for global weather and climate models based on machine learning},
	volume = {11},
	year = {2018}}

@article{chantry2021machinelearning,
	author = {Chantry, Matthew and Hatfield, Sam and Dueben, Peter and Polichtchouk, Inna and Palmer, Tim},
	doi = {10.1029/2021MS002477},
	journal = {Journal of Advances in Modeling Earth Systems},
	number = {7},
	pages = {e2021MS002477},
	title = {Machine Learning Emulation of Gravity Wave Drag in Numerical Weather Forecasting},
	url = {https://doi.org/10.1029/2021MS002477},
	volume = {13},
	year = {2021}}

@article{stewart2013data,
	author = {Stewart, Laura M and Dance, Sarah L and Nichols, Nancy K},
	doi = {10.3402/tellusa.v65i0.30587},
	journal = {Tellus A: Dynamic Meteorology and Oceanography},
	number = {1},
	pages = {19546},
	publisher = {Taylor \& Francis},
	title = {Data assimilation with correlated observation errors: experiments with a 1-D shallow water model},
	url = {https://doi.org/10.3402/tellusa.v65i0.19546},
	volume = {65},
	year = {2013}}

@article{desroziers2005diagnosis,
	author = {Desroziers, G{\'e}rald and Berre, Loic and Chapnik, Bernard and Poli, Paul},
	doi = {10.1256/qj.05.108},
	journal = {Quarterly Journal of the Royal Meteorological Society},
	pages = {3385--3396},
	publisher = {Wiley Online Library},
	title = {Diagnosis of observation, background and analysis--error statistics in observation space},
	url = {https://rmets.onlinelibrary.wiley.com/doi/10.1256/qj.05.108},
	volume = {131},
	year = {2005}}

@article{geer2019correlated,
	author = {Geer, Alan J.},
	doi = {10.5194/amt-12-3629-2019},
	journal = {Atmospheric Measurement Techniques},
	pages = {3629--3657},
	title = {Correlated observation error models for assimilating all-sky infrared radiances},
	url = {https://amt.copernicus.org/articles/12/3629/2019/},
	volume = {12},
	year = {2019}}

@article{JanjicEtAl,
	author = {Janji{\'c}, Tijana and Bormann, Niels and Bocquet, Marc and Carton, JA and Cohn, Stephen E and Dance, Sarah L and Losa, Svetlana N and Nichols, Nancy K and Potthast, Roland and Waller, Joanne A and others},
	doi = {10.1002/qj.3130},
	eprint = {https://rmets.onlinelibrary.wiley.com/doi/pdf/10.1002/qj.3130},
	journal = {Quarterly Journal of the Royal Meteorological Society},
	number = {713},
	pages = {1257-1278},
	title = {On the representation error in data assimilation},
	url = {https://rmets.onlinelibrary.wiley.com/doi/abs/10.1002/qj.3130},
	volume = {144},
	year = {2018}}

@article{Unet,
	author = {Olaf Ronneberger and Philipp Fischer and Thomas Brox},
	doi = {10.1007/978-3-319-24574-4_28},
	bibsource = {dblp computer science bibliography, https://dblp.org},
	eprint = {1505.04597},
	eprinttype = {arXiv},
	journal = {CoRR},
	title = {U-Net: Convolutional Networks for Biomedical Image Segmentation},
	url = {http://arxiv.org/abs/1505.04597},
	volume = {abs/1505.04597},
	year = {2015}}

@misc{rttov,
	address = {Shinfield Park, Reading},
	author = {Eyre, John R.},
	doi = {10.21957/xsg8d92y3},
	journal = {ECMWF Technical Memoranda},
	language = {eng},
	month = {03/1991},
	number = {176},
	pages = {28},
	publisher = {ECMWF},
	title = {A fast radiative transfer model for satellite sounding systems},
	url = {https://www.ecmwf.int/node/9329},
	year = {1991}}

@article{recurrent-residual-cnn,
	author = {Alom, Md Zahangir and Hasan, Mahmudul and Yakopcic, Chris and Taha, Tarek M and Asari, Vijayan K},
	doi={10.48550/arXiv.1802.06955},
	journal = {arXiv preprint arXiv:1802.06955},
	title = {Recurrent residual convolutional neural network based on u-net (r2u-net) for medical image segmentation},
	year = {2018}}

@inproceedings{smith2019super,
	author = {Smith, Leslie N and Topin, Nicholay},
	doi={10.1117/12.2520589},
	booktitle = {Artificial intelligence and machine learning for multi-domain operations applications},
	organization = {SPIE},
	pages = {369--386},
	title = {Super-convergence: Very fast training of neural networks using large learning rates},
	volume = {11006},
	year = {2019}}

@article{StegmannEtAl,
	author = {Patrick G. Stegmann and Benjamin Johnson and Isaac Moradi and Bryan Karpowicz and Will McCarty},
	doi = {10.1016/j.jqsrt.2022.108088},
	issn = {0022-4073},
	journal = {Journal of Quantitative Spectroscopy and Radiative Transfer},
	pages = {108088},
	title = {A deep learning approach to fast radiative transfer},
	url = {https://www.sciencedirect.com/science/article/pii/S0022407322000255},
	volume = {280},
	year = {2022}}

@article{ZhouEtAl,
	author = {Yongbo Zhou and Yubao Liu and Chao Liu},
	doi = {10.1016/j.jqsrt.2021.107891},
	issn = {0022-4073},
	journal = {Journal of Quantitative Spectroscopy and Radiative Transfer},
	pages = {107891},
	title = {A machine learning-based method to account for 3D Short-Wave radiative effects in 1D satellite observation operators},
	url = {https://www.sciencedirect.com/science/article/pii/S0022407321003836},
	volume = {275},
	year = {2021}}

@article{Scheck,
	author = {Leonhard Scheck},
	doi = {10.1016/j.jqsrt.2021.107841},
	issn = {0022-4073},
	journal = {Journal of Quantitative Spectroscopy and Radiative Transfer},
	pages = {107841},
	title = {A neural network based forward operator for visible satellite images and its adjoint},
	url = {https://www.sciencedirect.com/science/article/pii/S0022407321003344},
	volume = {274},
	year = {2021}}

@article{BaurEtAl,
	author = {Baur, Florian and Scheck, Leonhard and Stumpf, Christina and K{\"o}pken-Watts, Christina and Potthast, Roland},
	doi = {10.5194/amt-16-5305-2023},
	journal = {Atmospheric Measurement Techniques},
	number = {21},
	pages = {5305--5326},
	publisher = {Copernicus Publications G{\"o}ttingen, Germany},
	title = {A neural-network-based method for generating synthetic  \qty{1.6}{\um} near-infrared satellite images},
	volume = {16},
	year = {2023}
	}

@article{WeiserEtAl,
	author = {Weiser, Alan and Sergio E. and Zarantonello},
	doi = {10.2307/2007922},
	journal = {Mathematics of Computation},
	number = {181},
	pages = {189--196},
	title = {A Note on Piecewise Linear and Multilinear Table Interpolation in Many Dimensions},
	url = {https://doi.org/10.2307/2007922},
	volume = {50},
	year = {1988}}

@misc{jahrbuch2023,
	author = {{Deutscher Wetterdienst}},
	note = {Abgerufen am: 19.10.2025},
	title = {{Jahrbuch 2023 des Deutschen Wetterdiensts}},
	url = {https://www.dwd.de/DE/leistungen/jahresberichte_dwd/jahresberichte/2023.html?nn=511948l},
	urldate = {2025-10-19},
	year = {2023}}

@misc{rttov14svr,
	author = {{NWP SAF}},
	howpublished = {\url{https://nwp-saf.eumetsat.int/site/download/documentation/rtm/docs_rttov14/rttov14_svr.pdf}},
	note = {Accessed: 2025-10-19},
	title = {RTTOV v14 Science and Validation Report},
	url = {\url{https://nwp-saf.eumetsat.int/site/download/documentation/rtm/docs_rttov14/rttov14_svr.pdf}},
	urldate = {2025-10-19},
	year = {2025}}

@article{scheck2020,
	author = {Scheck, Leonhard and Weissmann, Martin and Bach, Liselotte},
	doi = {10.1002/qj.3840},
	eprint = {https://rmets.onlinelibrary.wiley.com/doi/pdf/10.1002/qj.3840},
	journal = {Quarterly Journal of the Royal Meteorological Society},
	number = {732},
	pages = {3165-3186},
	title = {Assimilating visible satellite images for convective-scale numerical weather prediction: A case-study},
	url = {https://rmets.onlinelibrary.wiley.com/doi/abs/10.1002/qj.3840},
	volume = {146},
	year = {2020}}

@article{CRTM2023,
	author = {Johnson, Benjamin T and Dang, Cheng and Stegmann, Patrick and Liu, Quanhua and Moradi, Isaac and Auligne, Thomas},
	doi = {10.1175/BAMS-D-22-0015.1},
	journal = {Bulletin of the American Meteorological Society},
	number = {10},
	pages = {E1817--E1830},
	publisher = {American Meteorological Society},
	title = {The Community Radiative Transfer Model (CRTM): Community-focused collaborative model development accelerating research to operations},
	url = {"https://journals.ametsoc.org/view/journals/bams/104/10/BAMS-D-22-0015.1.xml"},
	volume = {104},
	year = {2023}}

@article{eyre2020_1,
	author = {Eyre, John R. and English, Stephen J. and Forsythe, Mary},
	doi = {10.1002/qj.3654},
	eprint = {https://rmets.onlinelibrary.wiley.com/doi/pdf/10.1002/qj.3654},
	journal = {Quarterly Journal of the Royal Meteorological Society},
	number = {726},
	pages = {49-68},
	title = {Assimilation of satellite data in numerical weather prediction. Part I: The early years},
	url = {https://rmets.onlinelibrary.wiley.com/doi/abs/10.1002/qj.3654},
	volume = {146},
	year = {2020}}

@article{eyre2022_2,
	author = {Eyre, John R. and Bell, William and Cotton, James and English, Stephen J. and Forsythe, Mary and Healy, Sean B. and Pavelin, Edward G.},
	doi = {10.1002/qj.4228},
	eprint = {https://rmets.onlinelibrary.wiley.com/doi/pdf/10.1002/qj.4228},
	journal = {Quarterly Journal of the Royal Meteorological Society},
	number = {743},
	pages = {521-556},
	title = {Assimilation of satellite data in numerical weather prediction. Part II: Recent years},
	url = {https://rmets.onlinelibrary.wiley.com/doi/abs/10.1002/qj.4228},
	volume = {148},
	year = {2022}}

@article{weng2020,
	author = {Fuzhong Weng and Benjamin T. Johnson and Peng Zhang and Stephen English},
	doi = {10.1016/j.jqsrt.2020.106826},
	issn = {0022-4073},
	journal = {Journal of Quantitative Spectroscopy and Radiative Transfer},
	pages = {106826},
	title = {Preface for the special issue of radiative transfer models for satellite data assimilation},
	url = {https://www.sciencedirect.com/science/article/pii/S0022407320300017},
	volume = {244},
	year = {2020}}

@book{weng2017,
	author = {Fuzhong Weng},
	doi = {10.1002/9783527336289.fmatter},
	eprint = {https://onlinelibrary.wiley.com/doi/pdf/10.1002/9783527336289.fmatter},
	isbn = {9783527336289},
	publisher = {John Wiley \& Sons, Ltd},
	title = {Passive Microwave Remote Sensing of the Earth},
	url = {https://onlinelibrary.wiley.com/doi/abs/10.1002/9783527336289.fmatter},
	year = {2017}}

@article{HintonEtAl,
	author = {Hinton, Geoffrey E and Salakhutdinov, Ruslan R},
	description = {Reducing the dimensionality of data with neural ne...[Science. 2006] - PubMed Result},
	doi = {10.1126/science.1127647},
	journal = {Science},
	month = Jul,
	number = 5786,
	pages = {504-507},
	pmid = {16873662},
	title = {Reducing the dimensionality of data with neural networks},
	url = {http://www.ncbi.nlm.nih.gov/sites/entrez?db=pubmed&uid=16873662&cmd=showdetailview&indexed=google},
	volume = 313,
	year = 2006}

@inproceedings{MasciEtAl,
	author = {Masci, Jonathan and Meier, Ueli and Ciresan, Dan C. and Schmidhuber, J{\"u}rgen},
	doi = {10.1007/978-3-642-21735-7_7},
	booktitle = {ICANN},
	pages = {52--59},
	publisher = {Springer},
	title = {Stacked Convolutional Auto-Encoders for Hierarchical Feature Extraction},
	volume = {6791},
	year = {2011}}

@article{BadrinarayananEtAl,
	author = {Badrinarayanan, Vijay and Kendall, Alex and Cipolla, Roberto},
	doi = {10.1109/TPAMI.2016.2644615},
	journal = {IEEE Transactions on Pattern Analysis and Machine Intelligence},
	number = {12},
	pages = {2481-2495},
	title = {SegNet: A Deep Convolutional Encoder-Decoder Architecture for Image Segmentation},
	volume = {39},
	year = {2017}}

@book{Kalnay2003,
	author = {Kalnay, Eugenia},
	doi = {10.1256/00359000360683511},
	publisher = {Cambridge University Press},
	title = {Atmospheric Modeling, Data Assimilation and Predictability},
	year = {2003}}

@book{Durran2010,
	author = {Durran, Dale R.},
	doi = {10.1007/978-1-4419-6412-0},
	publisher = {Springer},
	title = {Numerical Methods for Fluid Dynamics},
	year = {2010}}

@article{Williamson2007,
	author = {Williamson, David L.},
	doi = {10.2151/jmsj.85B.241},
	journal = {Journal of the Meteorological Society of Japan},
	number = {B},
	pages = {241--269},
	title = {The evolution of dynamical cores for global atmospheric models},
	volume = {85},
	year = {2007}}

@article{LeutbecherPalmer2008,
	author = {Leutbecher, Martin and Palmer, Tim N.},
	doi = {10.21957/c0hq4yg78},
	journal = {Journal of Computational Physics},
	pages = {3515--3539},
	title = {Ensemble forecasting},
	volume = {227},
	year = {2008}}

@article{Palmer2019,
	author = {Palmer, Tim N.},
	doi = {10.1002/qj.3383},
	journal = {Quarterly Journal of the Royal Meteorological Society},
	pages = {12--24},
	title = {The ECMWF ensemble prediction system: Looking back (more than) 25 years and projecting forward 25 years},
	volume = {145},
	year = {2019}}

@article{Hunt2007,
	author = {Hunt, Brian R. and Kostelich, Eric J. and Szunyogh, Istvan},
	doi = {10.1016/j.physd.2006.11.008},
	journal = {Physica D},
	pages = {112--126},
	title = {Efficient data assimilation for spatiotemporal chaos: A local ensemble transform Kalman filter},
	volume = {230},
	year = {2007}}

@article{Bannister2017,
	author = {Bannister, Ross N.},
	doi = {10.1002/qj.2982},
	journal = {Quarterly Journal of the Royal Meteorological Society},
	pages = {607--633},
	title = {A review of operational methods of variational and ensemble-variational data assimilation},
	volume = {143},
	year = {2017}}

@article{Schraff2016,
	author = {Schraff, Christoph and Reich, Hendrik and Rhodin, Andreas and Schomburg, Annika and Stephan, Klaus and Perianez, Africa and Potthast, Roland},
	doi = {10.1002/qj.2748},
	eprint = {https://rmets.onlinelibrary.wiley.com/doi/pdf/10.1002/qj.2748},
	journal = {Quarterly Journal of the Royal Meteorological Society},
	number = {696},
	pages = {1453-1472},
	title = {Kilometre-scale ensemble data assimilation for the COSMO model (KENDA)},
	url = {https://rmets.onlinelibrary.wiley.com/doi/abs/10.1002/qj.2748},
	volume = {142},
	year = {2016}}

@book{NakamuraPotthast2015,
	author = {Nakamura, Gen and Potthast, Roland},
	doi = {10.1088/978-0-7503-1218-9},
	isbn = {978-0-7503-1218-9},
	publisher = {IOP Publishing},
	series = {2053-2563},
	title = {Inverse Modeling},
	url = {https://doi.org/10.1088/978-0-7503-1218-9},
	year = {2015}}

@article{Vobig2021,
	author = {Vobig, Klaus and Stephan, Klaus and Blahak, Ulrich and Khosravian, Kobra and Potthast, Roland},
	doi = {10.1002/qj.4157},
	eprint = {https://rmets.onlinelibrary.wiley.com/doi/pdf/10.1002/qj.4157},
	journal = {Quarterly Journal of the Royal Meteorological Society},
	number = {740},
	pages = {3789-3805},
	title = {Targeted covariance inflation for 3D-volume radar reflectivity assimilation with the LETKF},
	url = {https://rmets.onlinelibrary.wiley.com/doi/abs/10.1002/qj.4157},
	volume = {147},
	year = {2021}}

@article{VobigPotthast2025,
	author = {Vobig, Klaus and Potthast, Roland and Stephan, Klaus},
	doi = {10.5194/npg-32-471-2025},
	journal = {Nonlinear Processes in Geophysics},
	number = {4},
	pages = {471--488},
	title = {On process-oriented conditional targeted covariance inflation for 3D-volume radar data assimilation},
	volume = {32},
	year = {2025}}

@article{Potthast2022,
	author = {Potthast, Roland and Vobig, Klaus and Blahak, Ulrich and Simmer, Clemens},
	doi = {10.1175/MWR-D-21-0017.1},
	journal = {Monthly Weather Review},
	number = {5},
	pages = {969--980},
	title = {Data assimilation of nowcasted observations},
	volume = {150},
	year = {2022}}

@article{Zaengl2015,
	author = {Z{\"a}ngl, G{\"u}nther and Reinert, Daniel and R{\'\i}podas, Pilar and Baldauf, Michael},
	doi = {10.1002/qj.2378},
	eprint = {https://rmets.onlinelibrary.wiley.com/doi/pdf/10.1002/qj.2378},
	journal = {Quarterly Journal of the Royal Meteorological Society},
	number = {687},
	pages = {563-579},
	title = {The ICON (ICOsahedral Non-hydrostatic) modelling framework of DWD and MPI-M: Description of the non-hydrostatic dynamical core},
	url = {https://rmets.onlinelibrary.wiley.com/doi/abs/10.1002/qj.2378},
	volume = {141},
	year = {2015}}

@misc{icon-model-website,
	author = {{DWD and partners}},
	howpublished = {\url{https://www.icon-model.org}},
	note = {Accessed: 2025/12/30},
	title = {ICON Climate and Weather Prediction Model}}

@article{Dance2016,
	article-number = {581},
	author = {Waller, Joanne A. and Ballard, Susan P. and Dance, Sarah L. and Kelly, Graeme and Nichols, Nancy K. and Simonin, David},
	doi = {10.3390/rs8070581},
	issn = {2072-4292},
	journal = {Remote Sensing},
	number = {7},
	title = {Diagnosing Horizontal and Inter-Channel Observation Error Correlations for SEVIRI Observations Using Observation-Minus-Background and Observation-Minus-Analysis Statistics},
	url = {https://www.mdpi.com/2072-4292/8/7/581},
	volume = {8},
	year = {2016}}

@article{Potthast2019,
	author = {Potthast, Roland and Walter, Anne and Rhodin, Andreas},
	doi = {10.1175/MWR-D-18-0028.1},
	journal = {Monthly Weather Review},
	number = {1},
	pages = {345 - 362},
	title = {A localized adaptive particle filter within an operational NWP framework},
	url = {"https://journals.ametsoc.org/view/journals/mwre/147/1/mwr-d-18-0028.1.xml"},
	volume = {147},
	year = {2019}}

@article{Schomburg2026,
        author = {Schomburg, Annika and Schraff, Christoph and Köpken-Watts, Christina and Stephan, Klaus and Schühl, Liselotte and Faulwetter, Robin},
        title = {Operational all-sky assimilation of geostationary infrared water-vapour channels in the regional ICON-D2 model with an ensemble Kalman filter},
        journal = {Quarterly Journal of the Royal Meteorological Society},
        pages = {e70113},
	year = {2026},
        doi = {https://doi.org/10.1002/qj.70113},
        url = {https://rmets.onlinelibrary.wiley.com/doi/abs/10.1002/qj.70113}
}

\end{document}